\documentclass{article}
\usepackage[utf8]{inputenc}
\usepackage{amsmath,amssymb}
\usepackage{comment}

\usepackage{svg}

\usepackage[english]{babel}
\usepackage{blindtext}
\usepackage{amsthm}
\usepackage{graphicx}
\graphicspath{ {./images/} }
\usepackage{float}
\usepackage{caption}
\usepackage{subcaption}
\usepackage{lipsum}
\usepackage{mwe}
\usepackage{setspace}
\usepackage{url}

\DeclareMathAlphabet{\mathpzc}{OT1}{pzc}{m}{it}
\usepackage{mathrsfs}

\DeclareMathOperator{\EX}{\mathbb{E}}
\DeclareMathOperator{\Var}{Var}

\usepackage{gensymb}

\DeclareMathOperator\artanh{artanh}
\DeclareMathOperator\arcosh{arcosh}

\newtheorem{theorem}{Theorem}[section]
\newtheorem{lemma}[theorem]{Lemma}
\newtheorem{corollary}{Corollary}[theorem]

\theoremstyle{remark}

\theoremstyle{definition}

\def\keywordname{{\bfseries \emph Keywords}}%
\def\keywords#1{\par\addvspace\medskipamount{\rightskip=0pt plus1cm
\def\and{\ifhmode\unskip\nobreak\fi\ $\cdot$
}\noindent\keywordname\enspace\ignorespaces#1\par}}

\title{\textbf{Emergence of Minkowski-Spacetime
by Simple Deterministic Graph Rewriting}}
\author{Gabriel Leuenberger 
}
\date{2021\footnote{Our earliest draft was published in April 2020 \cite{leuenbergerApril}.}}

\begin{document}

\maketitle

\begin{abstract}
The causal set program as well as the Wolfram physics project leave open the problem of how a graph that is a (3+1)-dimensional Minkowski-spacetime according to its simple geodesic distances, could be generated solely from simple deterministic rules.
This paper provides a solution by describing simple rules that characterize discrete Lorentz boosts between 4D lattice graphs, which combine further to form Wigner rotations that produce isotropy and lead to the emergence of the continuous Lorentz group and the (3+1)-dimensional Minkowski-spacetime.
On such graphs, the speed of light, the proper time interval, as well as the proper length are all shown to be highly accurate.
\end{abstract}

\section{Introduction}
The Causal Set Program \cite{bombelli1987space,reid1999introduction,dribus2017discrete} and the recent Wolfram Physics Project \cite{wolfram2002new,wolfram2020class}, both seek to uncover the network of causal relations at the plank-scale, that is the fundamental structure of space-time. When zooming out to the macroscopic scale, this structure should at least manifest the following properties of Minkowski spacetime:\\
- (3+1)-dimensionality: One temporal and three spatial dimensions.\\
- Apparent continuity of space and time.\\
- Lorentz symmetry, which includes:\\
\hspace*{10pt} - Isotropy, i.e.: Rotational invariance.\\
\hspace*{10pt} - Accurate time dilation.\\
\hspace*{10pt} - Constancy of the maximal speed, i.e.: the speed of light.\\
\hspace*{10pt} - Euclidean distance (Pythagorean theorem can be derived).\\
In this paper we present novel discrete structures generated solely by local deterministic rules of remarkable simplicity, that fulfill all of the above requirements. 

This is different from previous approaches to Lorentz symmetry, such as randomly sprinkled causal sets, which date back to Bombelli \cite{bombelli1987space, bombelli2009discreteness, dribus2017discrete}. 
While randomly sprinkled causal sets were quite useful for study purposes, their construction process employed a preexisting continuous (3+1)-dimensional space, which becomes unnecessary with our new approach.

Previously, also Bolognesi \cite{bolognesi2013algorithmic, bolognesi2017spacetime} achieved Lorentz-symmetry deterministically, without presupposing a continuous space. However, his emerging space-times were restricted to 1+1 dimensions only.
Paradoxically, our approach can generate 2+1 or 3+1 dimensions, while in some ways being even simpler than Bolognesi's approach. This is because these additional spatial dimensions are of help when letting Lorentz-symmetry emerge.

A further previous approach to Lorentz symmetry by Gorard \cite{gorard2020some}, was to avoid the simplest distance measures and instead define a more sophisticated distance-measure, based on random walks and the Wasserstein transportation metric, which was quite useful for deriving aspects of general relativity with the Wolfram model. 
We instead provide concrete, fully described graphs that succeed at remaining Lorentz-symmetric, even under simple graph geodesic distance measures. We hope to inspire the further refinement of the general theories through these graphs.

In \mbox{Section \ref{minkowski}} we describe our main ideas and theorems in terms of directed graphs. 
Such structures could in principle then be reformulated and generated within the frameworks of causal sets, Wolfram models, pure lambda calculus, graph rewriting systems, and others.
We show algebraically that the properties of Minkowski spacetime emerge at the large scale.
We then show the resulting accuracy of the speed of light as well as the accuracy of the proper time interval.
However, before moving to \mbox{Section \ref{minkowski}}, it is helpful to firstly understand an analogous graph, from which only the two-dimensional Euclidean plane emerges. 
We provide this entry point in the following section.
\\
\\

\section{The Emergent Euclidean Plane}\label{Euclid}
This section is concerned with the emergence of the two-dimensional Euclidean plane from an undirected graph.
We describe its construction and show its geometric properties algebraically.
Note that this graph is not even a subgraph of our spacetime-graph from \mbox{Section \ref{minkowski}}. 
However, their construction procedures as well as their mathematical treatments are analogous, which is why we recommend to understand \mbox{Section \ref{Euclid}} here before moving to \mbox{Section \ref{minkowski}}.

\subsection{One Pair of Interlaced Lattice Graphs}\label{shortest}
The \textit{shortest path metric}, also known as the \textit{geodesic distance} or simply, the $distance$ $d(U,V)$ between two vertices $U$ and $V$ of a graph, is the least number of steps across its edges to travel from $U$ to $V$.
It is obvious that, while an infinite square lattice graph\footnote{Lattice graphs are also known as gird graphs or mesh graphs.}, at the large scale, can approximate $\mathbb{R}^2$, its geodesic distance will approximate the Manhattan distance instead of the desired Euclidean distance \cite{deza2009encyclopedia}. 
It is thus often assumed that such regular structures must be avoided and some irregular structure must be used instead. 
Our construction, however, involves multiple interlacing square lattices, each of which represents a different angle of orientation\footnote{A vaguely distantly related concept was developed for fluid dynamics simulations \cite{wolfram1986cellular}
\cite{chen2006recovery}.
}, which leads to a geodesic distance that converges to the Euclidean distance, as we will show in \mbox{Subsection \ref{multi}}. 

Here we firstly describe how to interlace only two infinite square lattice graphs with each other in order to form a graph that we call $\mathcal{E}_2$ .
Let the two infinite SLGs (square lattice graphs) be called $\mathscr{L}$ and $\mathscr{L}'$. 
They are both subgraphs of $\mathcal{E}_2$ . $\mathscr{L}$ and $\mathscr{L}'$ share vertices with each other, i.e.: there are some vertices that are both part of $\mathscr{L}$ as well as part of $\mathscr{L}'$; let these be called 'shared vertices'. 
We postulate two simple rules in what follows. 
To formulate our first rule, it is helpful to assign the four cardinal directions to the steps taken on the lattices and you may imagine instructing a taxicab through the rectilinear Manhattan. 
Note however that the cardinal directions assigned on $\mathscr{L}$ will \textit{not} be aligned with the cardinal directions assigned on $\mathscr{L}'$.\\
\\
\textbf{Rule 1:} For each shared vertex $A$ and for each cardinal direction $D$, there is a shared vertex $B$, such that both of the following two paths are correct:\\
-Path on $\mathscr{L }$: Starting at $A$, take two steps in direction $D$, \\ then take a \textit{right} turn and one step to arrive at $B$.\\
-Path on $\mathscr{L'}$: Starting at $A$, take two steps in direction $D$, \\ then take a \textit{left} turn and one step to arrive at $B$.\\
\\
\textbf{Rule 2:} For $\mathscr{L}$ and $\mathscr{L}'$, their shared vertices never neighbour each other.\\
\\
If these two rules are followed, then the graph $\mathcal{E}_2$ is obtained, that is illustrated in \mbox{Figure \ref{fig:2L}}:
\begin{figure}[H]
    \centering
    \includegraphics[scale = 0.3875]{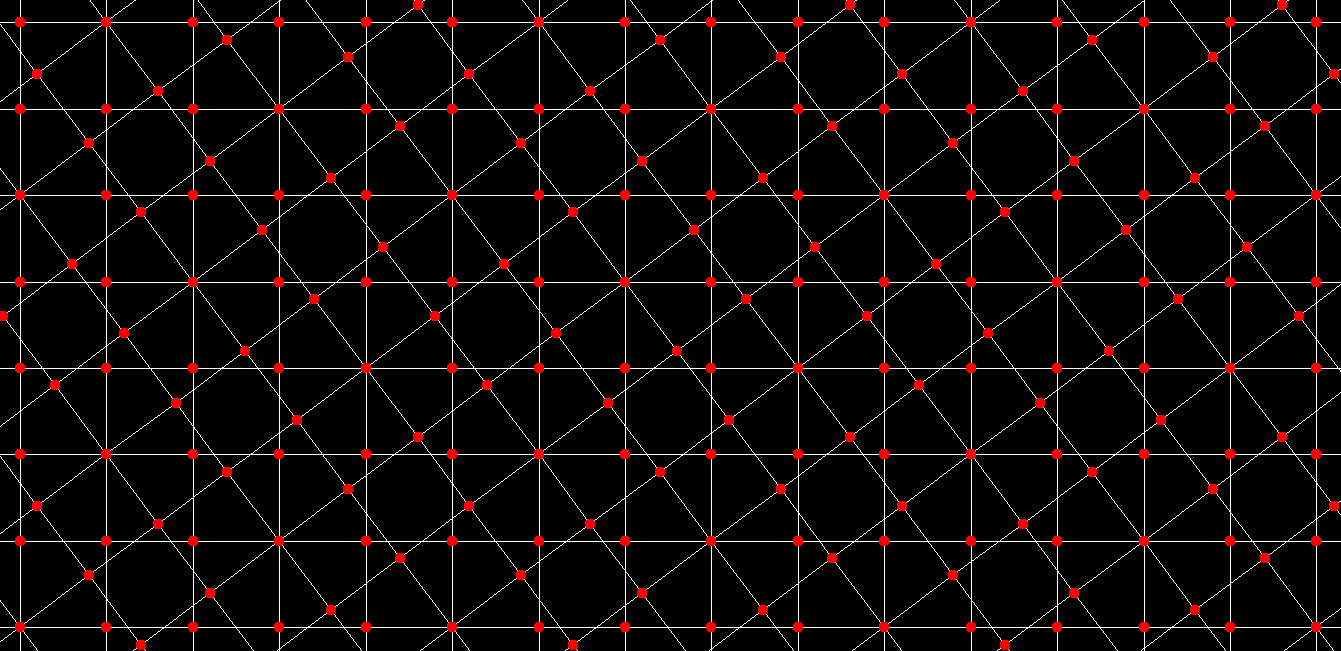}
    \caption{Graph $\mathcal{E}_2$ that was formed from two interlacing square lattice graphs by simple local rules. Shared vertices can be seen to have eight edges.}
    \label{fig:2L}
\end{figure}

It can easily be seen that each vertex of $\mathcal{E}_2$ either neighbours a shared vertex, or is itself a shared vertex, which will be of importance in the next subsection. 
Note that we now constructed this graph $\mathcal{E}_2$ solely from simple local rules without assigning coordinates to vertices and without performing arithmetics. 
We now will, however, start to assign coordinates to all vertices for study purposes and only later return to a coordinate-free formulation, that is in \mbox{Corollary \ref{voidOfCoordinates}} .
Firstly, we assign integer coordinates $(x,y) \in \mathbb{Z}^2$ to all vertices of $\mathscr{L}$ such that two vertices $P$ and $Q$ are connected by an edge exactly if $|\Vec{P}-\Vec{Q}|=1 $ . 
We can then use a second integer coordinate system for the vertices of $\mathscr{L'}$ that works identically. 
Note that shared vertices $A$ will have coordinates $\Vec{A}$ on $\mathscr{L}$ but will simultaneously also have different coordinates $\Vec{A}'$ on $\mathscr{L'}$; 
except for the central vertex $O$ which we define to have zero coordinates on both lattices, i.e.: 
$\Vec{O}=(0,0)=\Vec{O}'$. We can now reformulate the interlacing by using linear algebra. 
By applying rule 1 successively, we obtain the following equation for all shared vertices $A$:\\
\\
$\begin{bmatrix}
1 & 2\\
2 & -1
\end{bmatrix}
\Vec{A} \ \
= \ \
\begin{bmatrix}
-1 & 2\\
2 & 1
\end{bmatrix}
\Vec{A}' \ \
$\\
\\
This linear equation an then be rewritten equivalently as follows:
\\
$\Vec{A} \ \
= \ \
\frac{1}{5}
\begin{bmatrix}
3 & 4\\
-4 & 3
\end{bmatrix}
\Vec{A}' \ \
= \ \
\begin{bmatrix}
\cos(\theta) & \sin(\theta)\\
-\sin(\theta) & \cos(\theta)
\end{bmatrix}
\Vec{A}' \ \ ,
$\\
or alternatively:
\\
$\Vec{A}' \ \
= \ \
\frac{1}{5}
\begin{bmatrix}
3 & -4\\
4 & 3
\end{bmatrix}
\Vec{A} \ \
= \ \
\begin{bmatrix}
\cos(\theta) & -\sin(\theta)\\
\sin(\theta) & \cos(\theta)
\end{bmatrix}
\Vec{A} \ \ ,
$\\
\\
where:\ \
$\theta \ = \ 2 \arctan(\frac{1}{2}) \ = \ 0.9273.. \ = \ 53.13..^{\circ} $\\
\\
Note that these are rotation matrices with rotation angle $\theta$. 
So far, we only dealt with the integer-valued coordinates of the shared vertices, but we can in principle use the same rotation matrices to map all of the vertices of $\mathcal{E}_2$ onto the real-valued Euclidean plane $\mathbb{R}^2$, such that each edge corresponds to a Euclidean distance of one. 
Since such a mapping is possible, $\mathcal{E}_2$ is itself also a \textit{unit distance graph}. 
Note however, that one pair of SLGs is insufficient to further approximate Euclidean distance. 
In the following subsection, this is solved by extending $\mathcal{E}_2$ to an arbitrary number of interlaced lattice graphs.
\\

\subsection{Multitudinous Interlaced Lattice Graphs}\label{multi}
The previously used pair of lattice graphs $(\mathscr{L},\mathscr{L'})$ is now replaced by an ordered list of lattice graphs $[L,L',L'',..]$ . Let $n$ be the length of this list.
To construct a new graph $\mathcal{E}_n$, we re-use the same two rules that we introduced in the previous subsection. 
We apply these rules to each of the ordered pairs of neighbouring elements of the list.
Thus, for example, the rules must hold if we set $(\mathscr{L},\mathscr{L'})=(L,L')$, but must also hold if we set $(\mathscr{L},\mathscr{L'})=(L',L'')$, but must also hold if we set $(\mathscr{L},\mathscr{L'})=(L'',L''')$, and so forth.
To remove ambiguity, we further require that vertices are not shared between the lattice graphs unless required by the previous rules. 

As shown in the previous subsection, applying these rules corresponds to a rotation by the angle $\theta = 2 \arctan(\frac{1}{2}) $ . In our example, this would result in $L'''$ being rotated by $3\cdot\theta$ relative to $L$. For a list of length $n=5$, we illustrate such a construction in
\mbox{Figure \ref{fig:5L}}:
\\
\begin{figure}[H]
    \centering
    \includegraphics[scale = 0.345]{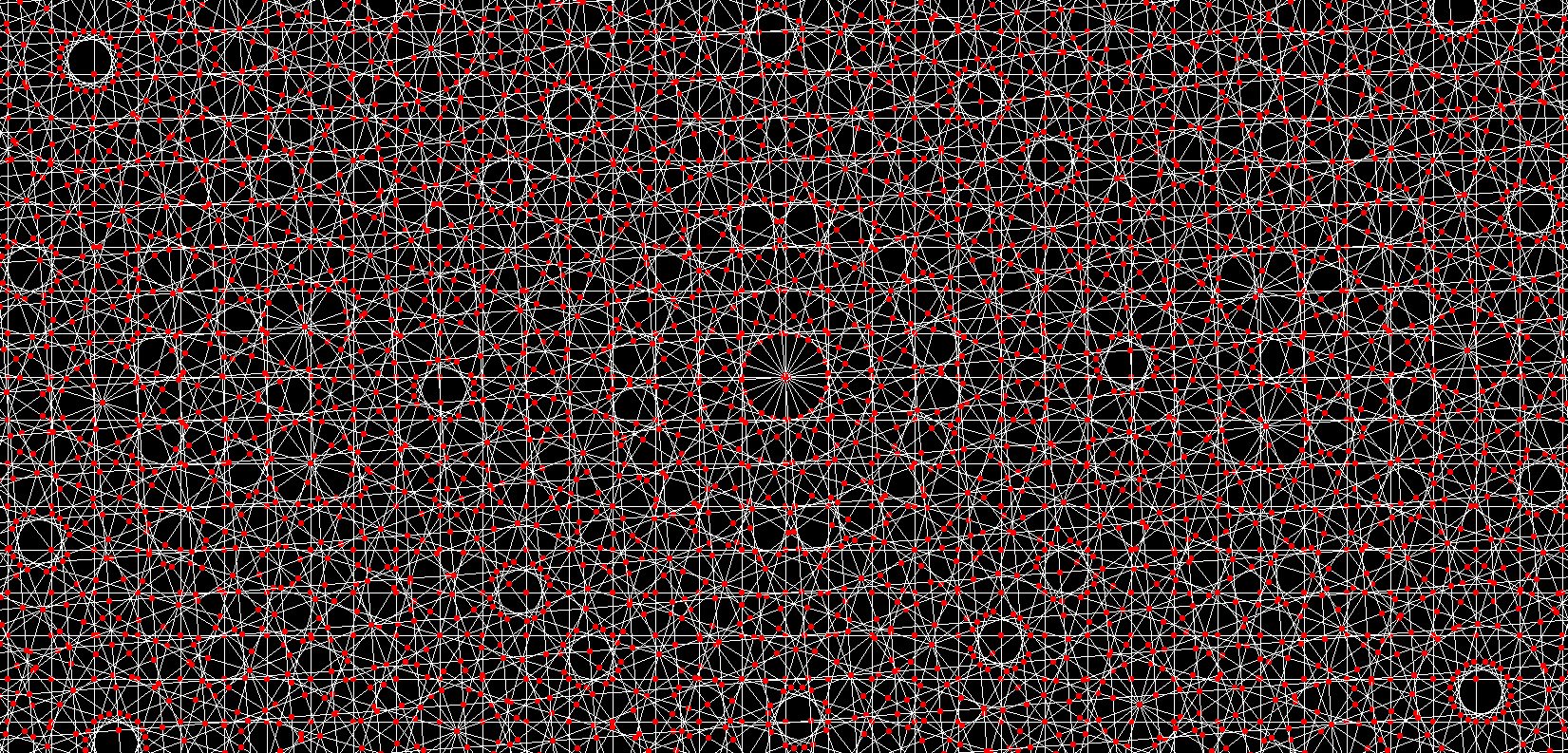}
    \caption{Graph $\mathcal{E}_n$ that was formed from a sequence of interlacing square lattice graphs by simple local rules, for $n=5$.}
    \label{fig:5L}
\end{figure}
Note that the resulting graph $\mathcal{E}_n$ remains a unit distance graph, regardless of the length of the list of SLGs. 
As is true for any unit distance graph, if we assign the corresponding real-valued coordinates $\Vec{U}, \Vec{V} \in \mathbb{R}^2$ to the vertices $U,V$, it follows that the geodesic distance between two vertices, that is the minimal number of steps between them, is greater or equal to the Euclidean distance $|\Vec{U}-\Vec{V}|$ according to their coordinates. 
While this provides us with a lower bound for the geodesic distance on $\mathcal{E}_n$, more interestingly, we shall derive a probabilistic \textit{upper} bound, or rather, the relative deviation of the geodesic distance from the Euclidean distance:
\\
\begin{theorem}\label{euclid}{\textbf{Accuracy of Euclidean Geodesic Distances on $\mathcal{E}_n$}:
}\\
For any sufficiently large $n \in \mathbb{N}$ and for randomly selected vertices $U,V$ of $\mathcal{E}_n$:\\
$
\textit{expected relative error} 
\ := \ \EX(| \frac{d(U,V)-|\Vec{U}-\Vec{V}|}{|\Vec{U}-\Vec{V}|} |)
\ < \ 
\frac{2 \pi^2}{n^2} + \frac{6n}{|\Vec{U}-\Vec{V}|}$
\\
\end{theorem}

\begin{proof}
Obviously, to each of $\mathcal{E}_n$'s SLGs, an angle of orientation can be assigned that is an integer multiple of $\theta$. 
One can now ask whether it is possible that, after multiple successive rotations by $\theta$, the initial orientation is restored, i.e.: 
In radians: 
Is there any $n \in \mathbb{N}$ such that $n\theta \equiv 0  \mod 2\pi $ ? 
Or equivalently: 
Is there any $n \in \mathbb{N}$ such that $\frac{n\theta}{2\pi} \in \mathbb{N}$ ? 
Note that the product between a natural number and an irrational number is always an irrational number. 
The answer to our question therefore depends on whether $\frac{\theta}{2\pi}$ is rational or not. 
The number in question $\frac{\theta}{2\pi} = \frac{2\arctan(\frac{1}{2})}{2\pi} = \frac{\arctan(\frac{1}{2})}{\pi}$ is called \textit{Plouffe's constant}\cite{plouffe1998computation}, and it was proven to be a transcendental number  by Margolius\cite{margolius2003plouffe}. 
It is hence also an irrational number, i.e.:
$\frac{\arctan(\frac{1}{2})}{\pi} \not\in \mathbb{Q}$ .
From this theorem it thus follows that successive rotations by angle $\theta$ will never return the orientation to its original state, and thus all orientations of the SLGs are different from each other. 
If we let the length $n$ of the list go towards infinity, we can furthermore use the \textit{equidistribution theorem}\cite{weyl1916gleichverteilung}; 
it implies that the angles of orientation of the lattice graphs are uniformly distributed on $[0,2\pi]$. 
This means that for sufficiently large $n$, the angles of orientation of the lattice graphs are approximately uniformly distributed on $[0,2\pi]$. 

We can use these approximately uniformly distributed directions to construct a path between two vertices $U$ and $V$ that will provide us with an upper bound for the geodesic distance $d(U,V)$. 
We are concerned with the case where the distance $d(U,V)$ is much larger than $n$.
Our heuristic approach to constructing a short path from $U$ to $V$ involves two SLGs, $L_i$ and $L_j$, that have orientations that are approximately aligned with the straight line from $U$ to $V$. 
Note that $L_i$ and $L_j$ are usually not neighbours in the list. 
To take the most direct path, the largest portion of the distance of the path is covered within these three approximately aligned SLGs. 
The path from $U$ to $V$ can be split into the following sequence of  five paths:
\\
(1) Shortest path from $U$ to $L_i$.\\
(2) Straight path within $L_i$ on line $l_i$.\\
(3) Shortest path from $l_i$ in $L_i$ to $l_j$ in $L_j$.\\
(4) Straight path within $L_j$ on line $l_j$.\\
(5) Shortest path from $L_j$ to $V$.\\
While the straight paths (2) and (4) are located within single SLGs, the other paths (1), (3), and (5) are not straight and lead through multiple different SLGs in as few steps as possible. 
Recall from the previous subsection that each vertex of $\mathcal{E}_2$ either neighbours a shared vertex, or is itself a shared vertex, which means that each vertex of $\mathcal{E}_2$ is at a distance from $L$ or $L'$ that is less than two. 
Similarly, it can easily be seen for $\mathcal{E}_n$ with $n=3$ that each vertex of $\mathcal{E}_n$ is at a distance from $L$, $L'$, or $L''$ that is less than three. 
And more generally, for all $n \in \mathbb{N} $ : 
For each SLG and each vertex of $\mathcal{E}_n$, the distance between the SLG and the vertex is less than $n$. 
Therefore, both paths (1) and (5) are of a length of less than $n$. Furthermore, the straight lines $l_i$ and $l_j$ cross such that the path (3) between them will also be of a length of less than $n$, analogously to (1) and (5). 
From these three equal upper bounds, we now get a first upper bound for the total length: $d(U,V)$ is less than or equal to $3n$ plus the length of (2) plus the length of (4). 
To further specify this bound, we now move on to estimate the lengths of the two straight paths (2) and (4). 
Let $\Vec{U}, \Vec{V} \in \mathbb{R}^2$ be the coordinate vectors assigned to the vertices $U$ and $V$ in the Euclidean plane and let $\overline{UV}$ be the straight line connecting them. 
The lengths of the two paths depend on how well the straight lines $l_i$ and $l_j$ are aligned with $\overline{UV}$, i.e.: 
It depends on how small the angles $\measuredangle(l_i,\overline{UV})$ and $\measuredangle(l_j,\overline{UV})$ are.
Let $\varepsilon$ be a value that is larger than both of these angles, i.e.: $|\measuredangle(l_i,\overline{UV})| \leq: \varepsilon :\geq |\measuredangle(l_j,\overline{UV})|$ . Consider a triangle that has two internal angles that are equal to $\varepsilon$.
Due to this symmetry, it will have two sides that are equally long and whose summed length will be equal to $\frac{d}{\cos(\varepsilon)}$, where $d$ is the length of the remaining side. It is then easy to see that a triangle that has two internal angles that are less than or equal to $\varepsilon$, will have two sides whose summed length will be less than or equal to $\frac{d}{\cos(\varepsilon)}$, where $d$ is again the length of the remaining side. 
Analogously, the sum of the lengths of the two paths (2) and (4) is less than $\frac{|\Vec{U}-\Vec{V}|}{\cos(\varepsilon)} +3n$ , where $3n$ have been added to account for the possible increase in the length of the paths (2) and (4) due to the changes in location when moving through each of the three paths (1), (3), and (5), that are each limited to a radius of $n$.
This is not to be confused with the previously obtained bound $3n$ on the sum of lengths of (1), (3), and (5) themselves. By combining both together, we obtain the bound shown in the following inequality:\\
$d(U,V) \ \leq \ \frac{|\Vec{U}-\Vec{V}|}{\cos(\varepsilon)} + 6n$
\\
We can now reshape this inequality to show the relative error of $d(U,V)$, i.e. the deviation of the geodesic distance from the Euclidean distance relative to the Euclidean distance itself:
\\
\textit{relative error} $\ := \ \frac{d(U,V)-|\Vec{U}-\Vec{V}|}{|\Vec{U}-\Vec{V}|} \ \leq \ \frac{1}{\cos(\varepsilon)} + \frac{6n}{|\Vec{U}-\Vec{V}|} - 1$
\\
Recall that in our heuristic approach, $L_i$ and $L_j$ are selected from the list of SLGs, such that their orientations allow for the angles $\measuredangle(l_i,\overline{UV})$, $\measuredangle(l_j,\overline{UV})$ to be minimal. The larger $n$ is, the more SLGs with different orientations there are available to choose from, which allows for a smaller $\varepsilon$ to exist, which implies a smaller relative error.
It now remains to be estimated how small this $\varepsilon$ could be depending on $n$. By observing the case of $n
=2$, we can see that a minimized $\varepsilon$ must be less than or equal to $\theta$, which then obviously also holds for all larger $n$. We therefore define the range of $\varepsilon$ as follows: $\forall n>1: \varepsilon :\in [0,\theta]$ . Within this range $[0,\theta]$ the following inequality holds: \ \
$\frac{1}{cos(\varepsilon)} \ \leq \ 1+\varepsilon^2 $ .\\
By combining this inequality with the previous relative error's inequality we get a new formula:\\
\textit{relative error} $\ \leq \ \varepsilon^2 + \frac{6n}{|\Vec{U}-\Vec{V}|}$
\\
We can use this simpler formula in combination with the variance identity $\Var(\varepsilon) = \EX(\varepsilon^2) - \EX(\varepsilon)^2$ to calculate the expected relative error:
\\
\textit{expected relative error} 
$\ \leq \ \EX( \ \varepsilon^2 + \frac{6n}{|\Vec{U}-\Vec{V}|} \ )
\ = \ \EX(\varepsilon)^2 + \Var(\varepsilon) + \frac{6n}{|\Vec{U}-\Vec{V}|}$
\\
In order to complete this estimate we now only need to calculate the expectation value $\EX(\varepsilon)$ and the variance $\Var(\varepsilon)$ of a minimal $\varepsilon$. 

Let $\varphi$ be an angle between the orientation of an SLG and the line $\overline{UV}$ . 
As previously discussed, the angles of orientation of SLGs are approximately uniformly distributed for sufficiently large $n$ . 
$\varphi$ can therefore be sampled from a uniform distribution over $[ -\pi , \pi ] $ . 
The probability $p$ that $\varphi$ will lie within an interval of size $\varepsilon$ will thus be equal to $\frac{\varepsilon}{2\pi} $ .
\\
$p \ := \ P(\varphi \in [0,\varepsilon] ) 
\ = \ \frac{ \varepsilon}{2 \pi}$
\\
Since $l_i$ is a line that points in one of the four orthogonal directions of an SLG, the probability that a minimized $|\measuredangle(l_i,\overline{UV})|$ will be smaller than or equal to $\varepsilon$, will be four times larger and furthermore doubled because we are taking the absolute value, thus resulting in a probability of $8p $ .
\\
$p_i \ := \ 
P(|\measuredangle(l_i,\overline{UV})| \leq \varepsilon ) 
\ = \ 8p \ = \ 
\frac{4 \varepsilon}{\pi}$
\\
For $l_j$ the situation is analogous, except that the angle $\measuredangle(l_j,\overline{UV})$ should have the correct sign for $l_j$ to intersect with $l_i$ such that $d(U,V)$ will be minimal, i.e: We are interested in the following probability that is halfed due the prescribed sign:
\\
$p_j \ := \ 
P(|\measuredangle(l_j,\overline{UV})| \in [0,\varepsilon] ) 
\ = \ \frac{p_i}{2} \ = \ 4p \ = \ 
\frac{2 \varepsilon}{\pi}$
\\
We can now use $p_i$ and $p_j$ to determine the probability $F$ that the path from $U$ to $V$ can be constructed by choosing $L_i$ out of the list of $n$ SLGs, and then choosing $L_j$ out of the remaining $(n-1)$ SLGs, given the constraint $\varepsilon$ :
\\
$F(\varepsilon) \ := \ 
(1-(1-p_i)^n)\cdot(1-(1-p_j)^{n-1})$
\\
$= (1-(1-\frac{4 \varepsilon}{\pi})^n)\cdot(1-(1-\frac{2 \varepsilon}{\pi})^{n-1})$\\
$\geq (1-(1-\frac{2 \varepsilon}{\pi})^n)\cdot(1-(1-\frac{2 \varepsilon}{\pi})^{n-1})$
\hspace*{\fill} ( \ for $\varepsilon \in [0,\frac{\pi}{4}] $ \ )
\\
$\geq (1-(1-\frac{2 \varepsilon}{\pi})^{n-1})^2
\ =: \ G(\varepsilon)$
\\
The smaller function $G$ was introduced in order to simplify the terms.
For the minimized $\varepsilon$, the derivatives then give us the pdfs (probability density functions) $f$ and $g$ : \\
$f(\varepsilon) \ := \ \frac{\partial}{\partial \varepsilon} 
F(\varepsilon)
$\\
$g(\varepsilon) \ := \ \frac{\partial}{\partial \varepsilon} 
G(\varepsilon)
$\\
Since all of these functions are monotone within our range of interest, and we know that $G$ is smaller than $F$, because $G$ increases slower than $F$, it follows that the pdf $g$ is more spread out than the pdf $f$ and therefore its expectation value is larger as well as its variance is larger, i.e.:\\
$\EX_f(\varepsilon) \ < \ \EX_g(\varepsilon)$ \ and \ $\Var_f(\varepsilon) \ < \ \Var_g(\varepsilon) $
\\
We now derive an upper bound for the expectation value:
\\
$\EX(\varepsilon) 
\ = \ \EX_f(\varepsilon) 
\ < \ \EX_g(\varepsilon) 
\ = \ 
\int_0^{\frac{\pi}{4}} \
\varepsilon
g(\varepsilon) \ d\varepsilon$
$\ = \ 
\int_0^{\frac{\pi}{4}} \
\varepsilon \frac{\partial}{\partial \varepsilon} 
G(\varepsilon) \ d\varepsilon 
\\
\ = \ 
\int_0^{\frac{\pi}{4}} \
\varepsilon \frac{\partial}{\partial \varepsilon} 
(1-(1-\frac{2 \varepsilon}{\pi})^{n-1})^2 \ d\varepsilon
\\ \ = \ 
\int_0^{\frac{\pi}{4}} \ 
\frac{4 \varepsilon}{\pi} 
(n-1)(1-\frac{2\varepsilon}{\pi})^{n-2}(1-(1-\frac{2\varepsilon}{\pi})^{n-1})
\ d\varepsilon$
\\
$\ \leq \ 
\int_0^{\frac{\pi}{4}} \
\frac{4 \varepsilon}{\pi} 
(n-1)(1-\frac{2\varepsilon}{\pi})^{n-2}
\ d\varepsilon$
\\
$\ < \ 
\int_0^{\frac{\pi}{2}} \
\frac{4 \varepsilon}{\pi} 
(n-1)(1-\frac{2\varepsilon}{\pi})^{n-2}
\ d\varepsilon$
\\
$\ = \ \frac{\pi}{n}$
\\
Next we calculate the variance:
\\
$\Var(\varepsilon) 
\ = \ \Var_f(\varepsilon) 
\ < \ \Var_g(\varepsilon) 
\ = \ \int_0^{\frac{\pi}{4}} \
(\varepsilon-\EX_g(\varepsilon))^2 g(\varepsilon) \ d\varepsilon$
\\
$< \ \int_0^{\frac{\pi}{4}} \
\varepsilon^2 g(\varepsilon) \ d\varepsilon$
\\
$< \ 
\int_0^{\frac{\pi}{2}} \
\frac{4 \varepsilon^2}{\pi} 
(n-1)(1-\frac{2\varepsilon}{\pi})^{n-2}
\ d\varepsilon$
\hspace*{\fill} ( \ By analogous steps to earlier. \ )
\\
$= \ \frac{\pi^2}{n^2+n}
\ < \ \frac{\pi^2}{n^2}$
\\
So to summarize:\\ $\EX(\varepsilon) \ < \ \frac{\pi}{n}$
 \ and \ 
$\Var(\varepsilon) \ < \ \frac{\pi^2}{n^2}$
\\
We can now insert these two values into our earlier relative error formula:
\\
\textit{expected relative error} 
$\ \ \leq \ \ 
\EX(\varepsilon)^2 + \Var(\varepsilon) + \frac{6n}{|\Vec{U}-\Vec{V}|}
\ \ < \ \
\frac{2 \pi^2}{n^2} + \frac{6n}{|\Vec{U}-\Vec{V}|}$
\\
\textbf{q.e.d.}
\end{proof}
{\ }\\
\\
From \mbox{Theorem \ref{euclid}}, the following few corollaries are easily obtained. These corollaries are kept more general than the theorem, such that they also apply to similar graphs that are mentioned in the following \mbox{Subsection \ref{alternatives}}.

Firstly, it is worth pointing out that the approximation of the Euclidean distance is perfect in the limit, when stated as follows:
\begin{corollary}
$\ \lim_{n \rightarrow \infty}$
$(\ \lim_{|\Vec{U}-\Vec{V}| \rightarrow \infty}
\textit{\ expected relative error}\ ) \ \ = \ 0$\\
\end{corollary}
We can use the Bachmann–Landau notation to characterize the limiting behavior of the deviation:
\begin{corollary}
$\textit{\ expected relative error}
\ \ = \ \ \mathcal{O}( \ \frac{1}{n^2}+\frac{n}{distance} \ )$\\
\end{corollary}
A common related critical question is whether a square's diagonal's length will equal $\sqrt{2}$ relative to the square's side's length $k$, to which the answer is of course yes, in the limit:
\begin{corollary} \
$\lim_{n \rightarrow \infty}$
$(\ \lim_{k \rightarrow \infty}
\frac{d(U,V)}{k} \ ) 
\ \ = \ \sqrt{2}$ ,\\
where $\Vec{U} = (k,0)$ and $\Vec{V} = (0,k)$.\\
\end{corollary}
Lastly, we can also formulate a corollary that is void of any coordinates:
\begin{corollary}\label{voidOfCoordinates}
For $d \gg n \gg 1$:\\
For any set $S$ of vertices of $\mathcal{E}_n$, where all the geodesic distances between these vertices are larger than a constant $d$ : 
There exists a set of points in the Euclidean plane, such that the set of ratios between the Euclidean distances between these points is identical to the set of ratios between the geodesic distances between the vertices in $S$.\\
\end{corollary}

\subsection{Variations, Generalisations, and Alternatives}\label{alternatives}

Some alternatives to the aforedescribed graph $\mathcal{E}_n$ that are variations of the same concept, are briefly discussed here in order to provide a more generalized picture.

\subsubsection{alternative angles}
The angle $\theta$ was determined by the pair of numbers $(2,1)$ of steps taken in different directions within Rule 1. This pair of integers was chosen for its simplicity but could otherwise have been chosen arbitrarily, as long as the two integers were not equal to each other nor equal to zero; our corollaries would still be holding then. 
This is because our proof of the \mbox{Theorem \ref{euclid}} is based on the irrationality of Plouffe's constant $\frac{\arctan(\frac{1}{2})}{\pi} \not\in \mathbb{Q}$ , where the ratio $\frac{1}{2}$ appears, that can be generalized to other ratios $q$ , i.e.:
$\frac{\arctan(q)}{\pi} \not\in \mathbb{Q}$ , where $q \in \mathbb{Q}$ and $q \not\in \{-1,0,1\}$,
as proven by Smith \cite{smith2003pythagorean}. Hence any angle of the form $\theta = 2\arctan(q)$ would be admissible with SLGs.

\subsubsection{Alternative lattice graphs}
While we only employed square lattice graphs, other obvious choices are the hexagonal as well as the triangular lattice graphs. 
Our corollaries also apply when triangular lattice graphs are used instead of the SLG, since both allow for straight line paths, whereas for the hexagonal lattice, the geodesic distance would have to be multiplied with a correction factor of $\frac{2}{\sqrt{3}}$, in order to account for the absence of straight line paths. 
A further possibility is to use square lattice graphs, but where the rules are altered such that the SLGs are interlaced in such a way that each square represents a rhombus rather than a square, while still maintaining the unit distance graph property, as well as our corollaries.

\subsubsection{Non-unit distance graphs}
All graphs that we discussed so far were unit distance graphs.
This unit distance property is however unnecessary for a scaled version of \mbox{Theorem \ref{euclid}} to hold. 
Examples of such non-unit distance graphs can be obtained as variations of $\mathcal{E}_n$ by cancelling Rule 2 and replacing it with a simple rule that lets the two SLGs share more vertices with each other. 
Each of these additional shared vertices has to correspond to a pair of close-by vertices in $\mathcal{E}_n$.
Geodesic distances then become shorter than the Euclidean distance, while a lower bound, proportional to the Euclidean distance, remains, and thus, an accordingly scaled version of \mbox{Theorem \ref{euclid}} persists.

\newpage

\section{The Emergent Minkowski Spacetime}\label{minkowski}

In this section we construct an example of a GRIDS, which is our acronym for a directed \textbf{G}raph that is \textbf{R}elativistic, \textbf{I}sotropic, \textbf{D}eterministic, and \textbf{S}imple (GRIDS). 
Our example of a GRIDS is fully characterized by simple rules describing its local network structure, and yet, at the large scale it does yield a complete approximation of the continuous (3+1)-dimensional Minkowski spacetime\cite{minkowski1909raum} of special relativity theory\cite{einstein1905elektrodynamik} including the continuous hyperbolic space of the Lorentz group \cite{poincare1906dynamics,lorentz1904electromagnetic}. 
Outside the scope of this paper are other GRIDS, characterized by even fewer rules, that might be more difficult to conceive of.
Our example, $\mathcal{M}_n$, serves as a proof of concept for GRIDS. 
It features both light-like edges and time-like edges and is constructed with an emphasis on ease of coordinatization and ease of understanding.
\\
\subsection{Single Frame-Grid}\label{framegrid}
We firstly introduce the concept of a \textit{frame-grid}, which is a repeatedly occurring subgraph of a GRIDS. 
A frame-grid is a lattice graph, that corresponds to a single inertial frame of reference of the emergent Minkowski spacetime of a GRIDS. 
Simple repetitive rules govern how frame-grids are interlaced with each other in order to form a GRIDS. 

After this broad informal definition of a frame-grid, we now proceed to describing a specific frame-grid-example, $\mathscr{F}$, that we will use to construct our GRIDS-example, $\mathcal{M}_n$, in the following subsections.
$\mathscr{F}$ is an infinite directed graph and is also a four-dimensional lattice graph.
Note that $\mathscr{F}$ is \textit{not} the vertex-edge graph of a hypercubic honeycomb. 
$\mathscr{F}$ can easily be understood when integer coordinates $(x,y,z,t) \in \mathbb{Z}^4$ are assigned to each of its vertices. 
The set of vertices of $\mathscr{F}$ corresponds to the subset of $\mathbb{Z}^4$ where the sum $x+y+z+t$ is an even number, 
i.e.: The set
$\{(x,y,z,t) \in \mathbb{Z}^4 \ | \ x+y+z+t \equiv 0 \ ( mod \ 2 ) \}$.

Let us now define the \textit{light-like edges} of $\mathscr{F}$. 
Each vertex is the origin of six light-like edges directed away from it, leading in six different directions, to six other vertices.
Following such a directed edge always leads to a vertex where $t$ is increased by one, while exactly one of the three other coordinates $x,y,z$ is changed; 
it can change by $\pm 1$, which are two possible values, and hence the \textit{six} directions. 
Therefore each vertex also has six light-like edges directed at it, originating from six other vertices. \\
Let these six directions be labeled $\textsc{x}_+$ , $\textsc{x}_-$ , $\textsc{y}_+$ , $\textsc{y}_-$ , $\textsc{z}_+$ , and $\textsc{z}_-$ .

We now also define the \textit{time-like edges} of $\mathscr{F}$. 
Each vertex is the origin of one time-like edge directed away from it, leading to a vertex, where $t$ is increased by 2, while the other coordinates, $x,y,z$, stay unchanged. 
Therefore each vertex also has one time-like edge directed at it, that originates from another vertex.

Due to the many regularities, we take it as a given, that such lattice graphs $\mathscr{F}$ can be constructed without coordinates, solely through simple graph-rewriting rules, that we will not bother describing in this paper. 
These rules can easily be made to tag all light-like edges with their corresponding direction labels.
In the following subsection, we will be using these direction labels to denote steps from one vertex to another, along single light-like edges. 
See the following four examples of our notation of steps along single light-like edges and their associated movements in coordinates $(x,y,z,t)$:
\\
$+\textsc{x}_+ \ \widehat{=} \ (+1,0,0,+1)$
\\
$-\textsc{x}_+ \ \widehat{=} \ (-1,0,0,-1)$
\\
$+\textsc{x}_- \ \widehat{=} \ (-1,0,0,+1)$
\\
$-\textsc{x}_- \ \widehat{=} \ (+1,0,0,-1)$
\\
\\
\subsection{Interlaced Pair of Frame-Grids}

In this subsection, we describe how two frame-grids, $\mathscr{F}$ and $\mathscr{F}'$, are interlaced with each other in order to form $\mathcal{G}$, a graph that is a helpful intermediate step before understanding the GRIDS $\mathcal{M}_n$ .
Both frame-grids, $\mathscr{F}$ and $\mathscr{F}'$ , are subgraphs of $\mathcal{G}$, which represents an 'elementary' Lorentz transformation between their two inertial frames of reference. 
$\mathcal{G}$ is also a repeatedly occurring subgraph of $\mathcal{M}_n$.
\\
\subsubsection{Primitive Local Rules}\label{4Rules}
$\mathscr{F}$ and $\mathscr{F}'$ share vertices with each other, i.e.: 
There are some vertices that are both part of $\mathscr{F}$ as well as part of $\mathscr{F}'$; let these be called '\textit{shared vertices}'.
These shared vertices are arranged in a regular fashion, as characterized by the following rules.
We denote a step along a light-like edge of $\mathscr{F}$ as described in the previous subsection, and we denote a step along a light-like edge of $\mathscr{F'}$ identically, but with a stroke.\\
\\
\textbf{Rule 1:} 
$\mathscr{F}$ and $\mathscr{F}'$ share a vertex $O$ . Their other vertices are not shared unless required by the following rules.\\
\\
\textbf{Rule 2:} 
For all shared vertices $A$ :
\[ A + \textsc{x}_+ + \textsc{x}_+ = A + \textsc{x}_+' \]
\[ A + \textsc{x}_- = A + \textsc{x}_-' + \textsc{x}_-' \]
\\
\textit{Explanation}: From the strokes it can be seen that the left sides of the equations denote paths through $\mathscr{F}$ , while the right sides denote paths through $\mathscr{F'}$.
Therefore, for instance, if $A$ is a shared vertex, then $( A + \textsc{x}_+ + \textsc{x}_+ )$ is also a shared vertex. 
By successive application of Rule 2 , the graph depicted in \mbox{Figure \ref{fig:2F}} is obtained:
\\
\\
\begin{figure}[H]
    \centering
    \includegraphics[scale=0.45]{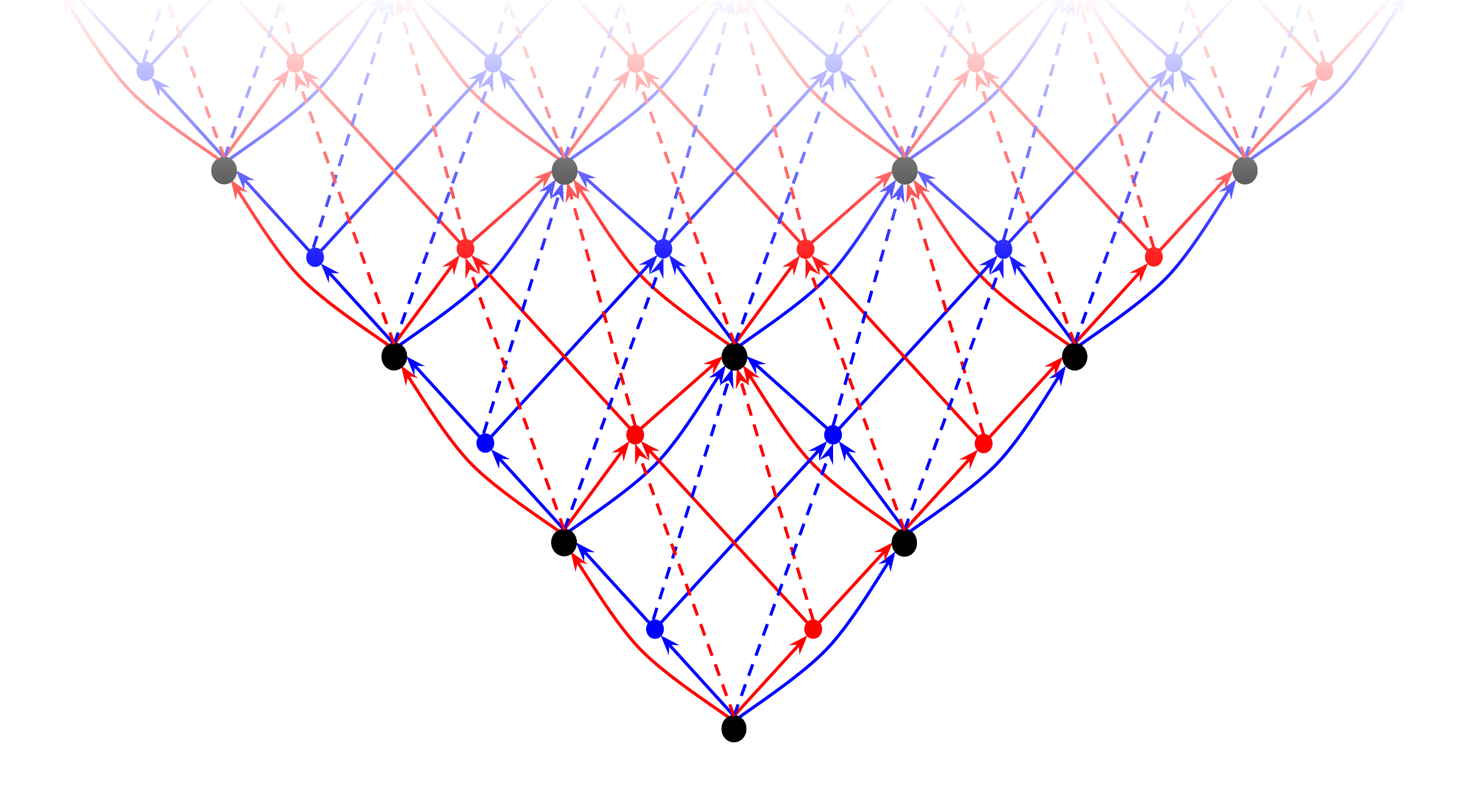}
    \caption{Part of the graph formed by the repeated application of Rule 2 .\\
    $\mathscr{F}$ is shown in red, $\mathscr{F'}$ is shown in blue, and their shared vertices are shown in black. 
    The solid arrows represent light-like edges while the dashed arrows represent time-like edges.    }
    \label{fig:2F}
\end{figure}
So far we interlaced two 2-dimensional subgraphs of $\mathscr{F}$ and $\mathscr{F}'$ to form the graph, depicted in \mbox{Figure \ref{fig:2F}} . 
Combined with the following rule, Rule 3, the interlacing is extended to four dimensions: \\
\\
\textbf{Rule 3:} 
For all shared vertices $A$ :
\[ A + \textsc{y}_+ - \textsc{z}_+ = A + \textsc{y}_+' - \textsc{z}_+' \]
\[ A + \textsc{y}_- - \textsc{z}_+ = A + \textsc{y}_-' - \textsc{z}_+' \]
\\
\textit{Explanation}: From the strokes it can again be seen that the left sides of the equations denote paths through $\mathscr{F}$ , while the right sides denote paths through $\mathscr{F'}$ . 
Note that the paths in Rule 3 move once forwards and once backwards in time, thus not changing the position in time overall. This rule simply copies the interlacing of the subgraph seen in \mbox{Figure \ref{fig:2F}} onto the many parallel subgraphs that have different $y$ and $z$ positions.

The repeated application of Rule 2 propagates only forwards in time. 
Analogously, Rule 3 does not propagate in all directions. 
We therefore add the following final rule, in order to propagate the interlacing into all times and all directions, for completeness sake:\\
\\
\textbf{Rule 4:} 
The previous rules are also apply with all of their paths reversed.\\
\\
\subsubsection{Reformulation through Coordinates}
While our set of rules was based on individual steps along edges, we now reformulate this set of rules using integer coordinates $(x,y,z,t) \in \mathbb{Z}^4$ that can be assigned to the vertices of $\mathscr{F}$ as we described in \mbox{Subsection \ref{framegrid}}. \\
Let $(x',y',z',t') \in \mathbb{Z}^4$ be the integer coordinates assigned on $\mathscr{F}'$ .

These eight coordinates belong to the same vertex if and only if the following four equations hold:
\[t + x \ = \ 2 \ ( t' + x' )\]
\[2 \ ( t - x ) \ = \ t' - x'\]
\[y = y'\]
\[z = z'\]
Note that these equations are sufficient to fully replace our four primitive rules, i.e.: These equations already fully describe $\mathcal{G}_2$ .

We further translate these equations into linear algebra. 
Let $\Vec{A}$ and $\Vec{A}'$ be the coordinates on $\mathscr{F}$ and $\mathscr{F}'$ , respectively, in the form of column vectors. The following linear equation is then equivalent to the previous four equations:
\[\begin{bmatrix}
1 & 0 & 0 & 1\\
-2 & 0 & 0 & 2\\
0 & 1 & 0 & 0\\
0 & 0 & 1 & 0\\
\end{bmatrix}
\Vec{A} \ \
= \ \
\begin{bmatrix}
2 & 0 & 0 & 2\\
-1 & 0 & 0 & 1\\
0 & 1 & 0 & 0\\
0 & 0 & 1 & 0\\
\end{bmatrix}
\Vec{A}' \ \
\]
This linear equation can then be rewritten equivalently as follows:
\\
\[
\Vec{A} \ \
= \ \ \frac{1}{4} \
\begin{bmatrix}
5 & 0 & 0 & 3\\
0 & 4 & 0 & 0\\
0 & 0 & 4 & 0\\
3 & 0 & 0 & 5\\
\end{bmatrix}
\Vec{A}'
\]
\[
 = \ \
\begin{bmatrix}
\cosh ( \ln (2) ) & 0 & 0 & \sinh(\ln(2))\\
0 & 1 & 0 & 0\\
0 & 0 & 1 & 0\\
\sinh(\ln(2)) & 0 & 0 & \cosh(\ln(2))\\
\end{bmatrix}
\Vec{A}' \ \
\]
\\
We can immediately see that this matrix represents a Lorentz transformation without rotation, thus called a Lorentz boost, with the following values:
\\
Velocity : \ $v = \frac{3}{5}c$ \ ;\ \ \ \
Rapidity : \ $w = \ln(2)$ \ ;\ \ \ \
Lorentz factor : \ $\gamma = \frac{5}{4}$ \ .
\\
The relation between these three physical quantities is the following:
\[ \artanh(v/c) = w = \arcosh(\gamma) \]
\\
\subsection{Multitudinous Interlaced Frame-Grids}
We now describe the construction of the GRIDS example $\mathcal{M}_n$ and then go on to calculate the accuracy of both, the speed of light as well as the proper time interval.
The construction is most easily shown visually by using the conformal disk model\footnote{It is also known as \textit{Poincaré disk model}, although originally discovered by Beltrami \cite{beltrami1868teoria} .} of the hyperbolic plane, to represent the relative rapidities and the angles between Lorentz boosts. 
This hyperbolic space later emerges naturally from many successive Lorentz boosts due to repeated interlacing.
In the previous subsection, we described how a frame-grid can be interlaced with another frame-grid, resulting in a Lorentz boost in the direction of dimension $x$, now depicted in the first disk of \mbox{Figure \ref{fig:nine disks}.} 
The second disk shows a frame-grid in the center that is interlaced analogously with other frame-grids, but in different perpendicular directions. 
These boosts in different perpendicular directions can easily be achieved by permuting the directions within rule 2 and rule 3 of the previous subsection accordingly. 
Note that these disks are only 2D cross sections of the 3D Poincaré ball model, where there are six such perpendicular directions.
Let $\mathcal{M}_1$ be the graph consisting of a central frame-grid that is interlaced with frame-grids in all six of these perpendicular directions, totalling a number of seven frame-grids.
When there is a sequence of interlaced frame-grids, we can assign a different Lorentz transformation to each frame-grid through the corresponding successive applications of the previously described coordinate transformation. 
Note that this graph has a special property that resembles the unit-distance graph property, i.e.: 
The graph can be embedded in a Minkowski space-time such that all the time-like edges correspond to time-like paths of the same unit time, while all the light-like edges will correspond to light-like paths. 
This property is automatically retained by $\mathcal{M}_n$ for all $n \in \mathbb{N}$.
Let $\mathcal{M}_2$ then be the graph consisting of $\mathcal{M}_1$, where all frame-grids are interlaced with further frame-grids in all unoccupied perpendicular directions, resulting in the 2D cross section depicted in the third disk and totalling a number of 37 frame-grids.
Note that, since this is a hyperbolic space, no square was formed, even-though the angles are perpendicular and all lines are of the same length as well as straight.
In the third disk we can furthermore start to observe the \textit{Wigner rotations} \cite{thomas1926motion,wigner1939unitary} caused by successive Lorentz boosts in different directions.
The Wigner rotation is a consequence of special relativity that is similarly astonishing to the \mbox{twin paradox}.
While the twin paradox concerns the time difference caused by successive Lorentz-boosts, the Wigner rotation concerns the change in orientation caused by successive Lorentz-boosts.
We obtain $\mathcal{M}_3$ by repeating the same procedure and so forth; this is also how we define $\mathcal{M}_n$ recursively for all $n \in \mathbb{N}$. 
The remaining disks visualize a few more of these steps.
\\
\begin{samepage}
\begin{figure}[H]
\captionsetup[subfigure]{labelformat=empty}
     \centering
     \begin{subfigure}[b]{0.3\textwidth}
         \centering
         \includegraphics[scale=0.34]{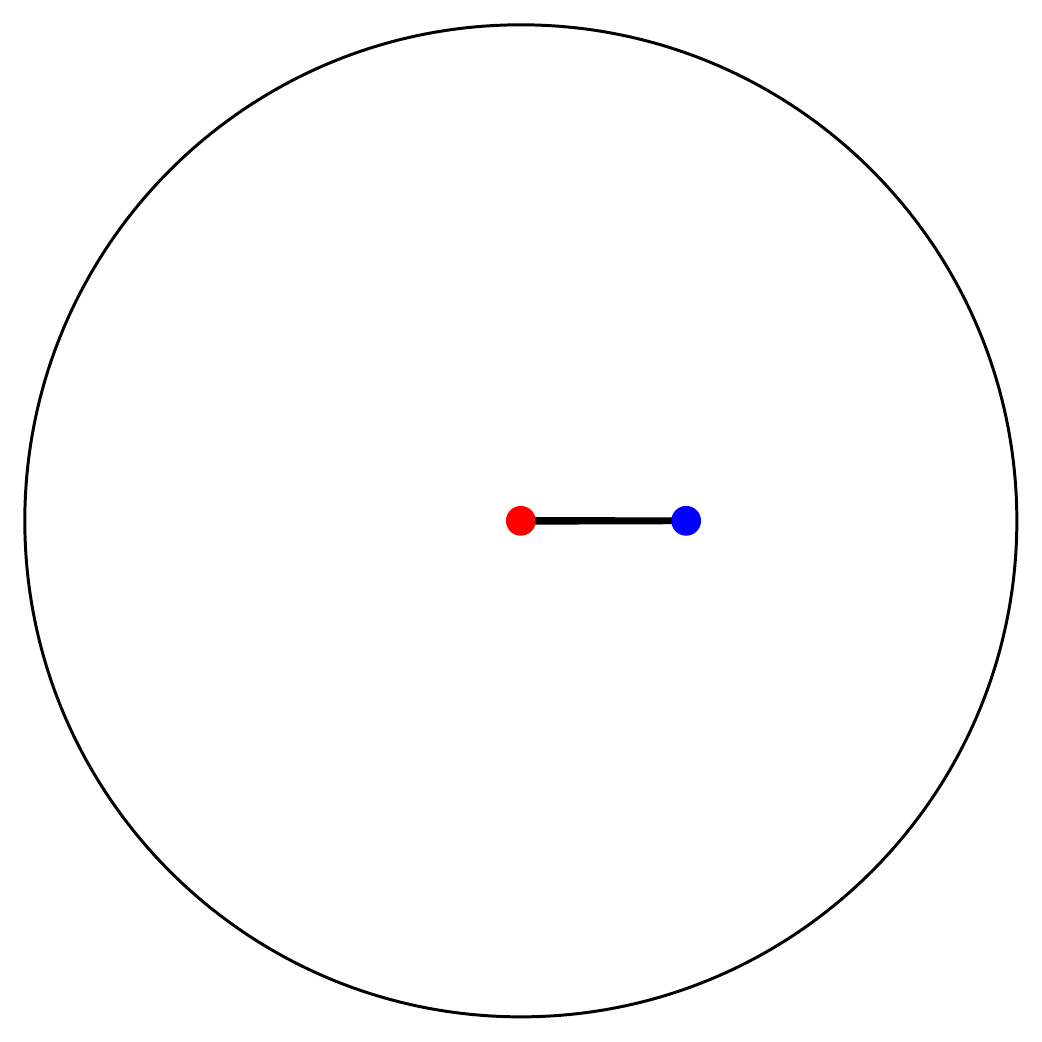}
         \caption{$\mathcal{G}$: $\mathcal{F}$ in red, $\mathcal{F}'$ in blue.}
     \end{subfigure}
     \hfill
     \begin{subfigure}[b]{0.3\textwidth}
         \centering
         \includegraphics[scale=0.17]{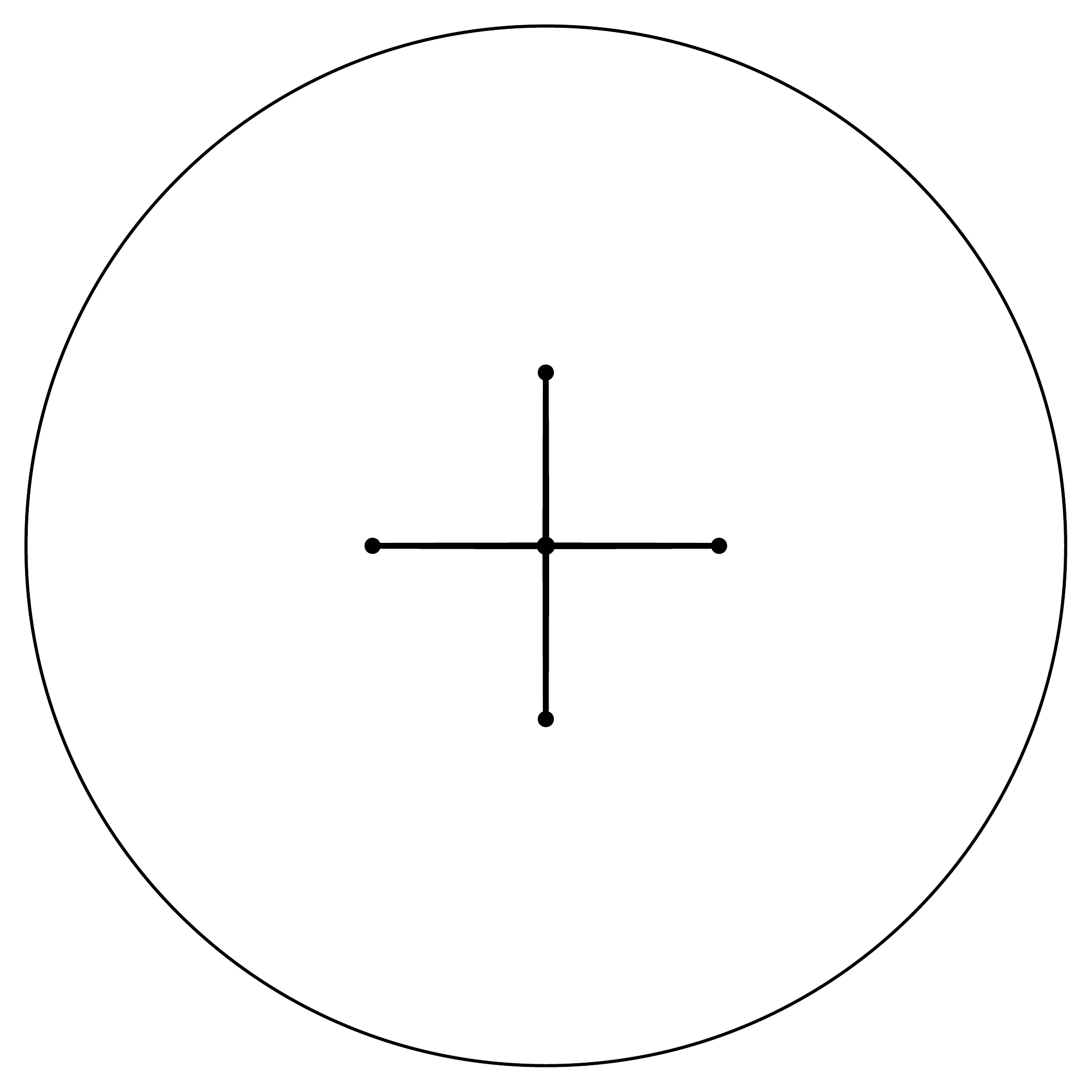}
         \caption{$\mathcal{M}_1$}
     \end{subfigure}
     \hfill
     \begin{subfigure}[b]{0.3\textwidth}
         \centering
         \includegraphics[scale=0.17]{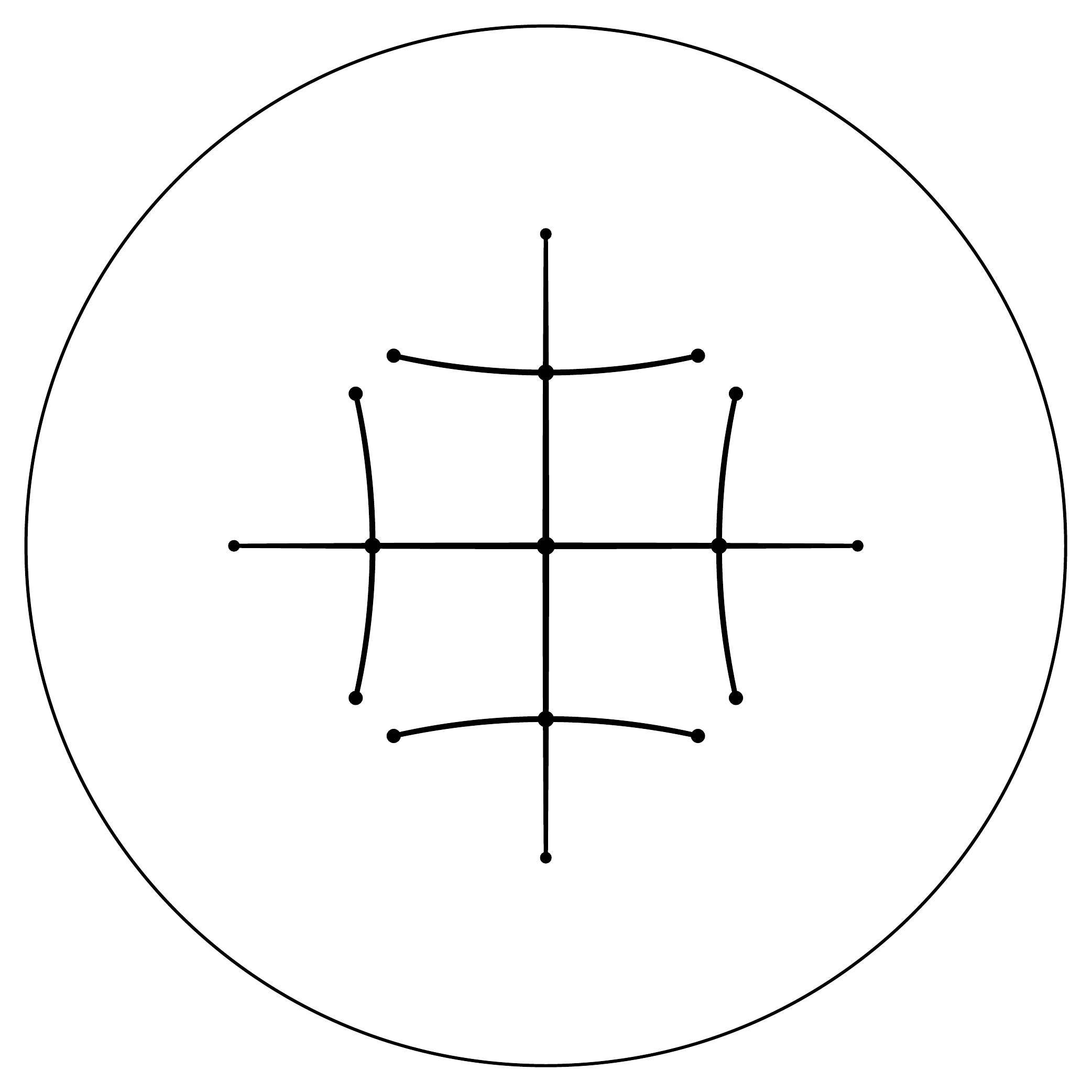}
         \caption{$\mathcal{M}_2$}
     \end{subfigure} 
     \par\bigskip
     \begin{subfigure}[b]{0.3\textwidth}
         \centering
         \includegraphics[scale=0.17]{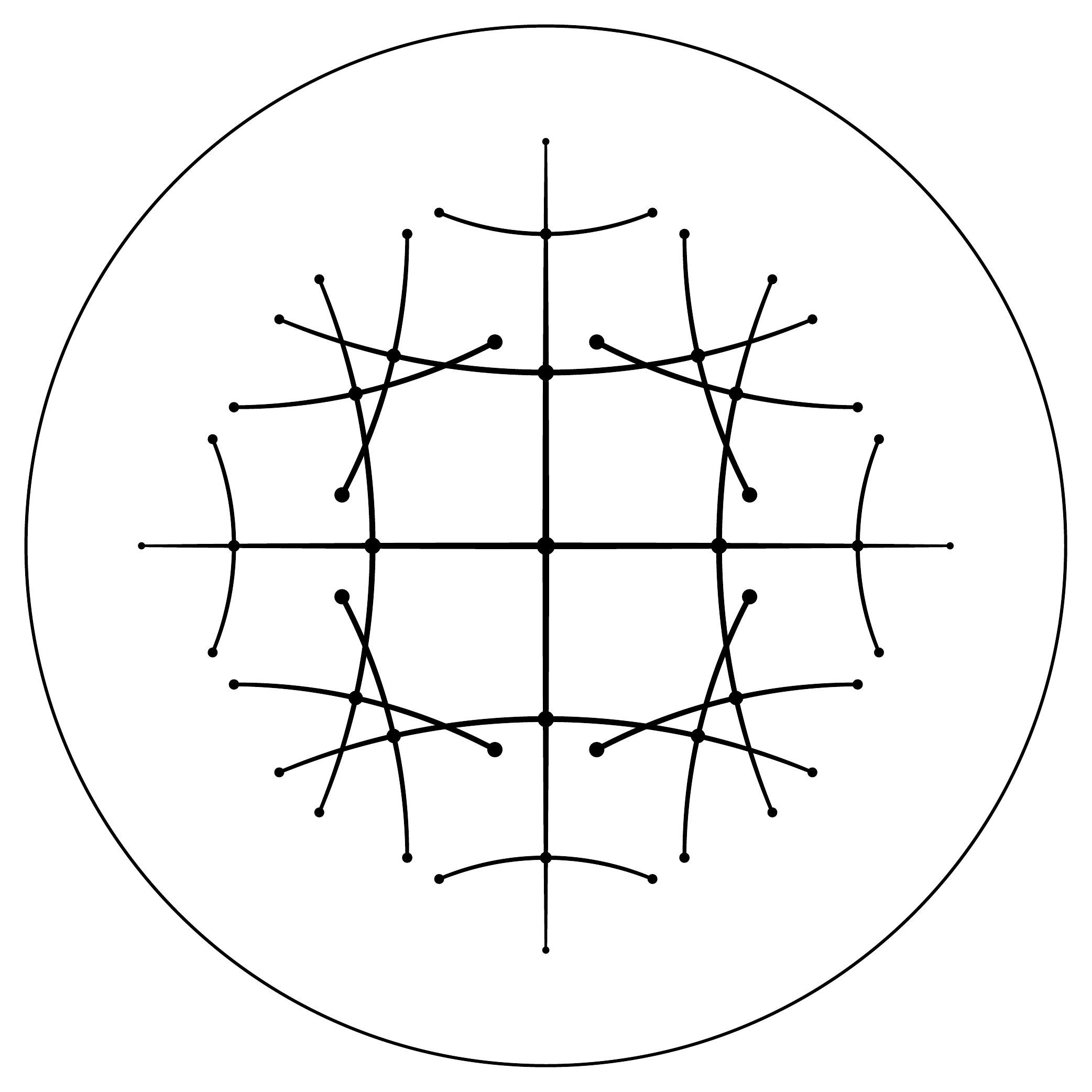}
         \caption{$\mathcal{M}_3$}
     \end{subfigure}
     \hfill
     \begin{subfigure}[b]{0.3\textwidth}
         \centering
         \includegraphics[scale=0.17]{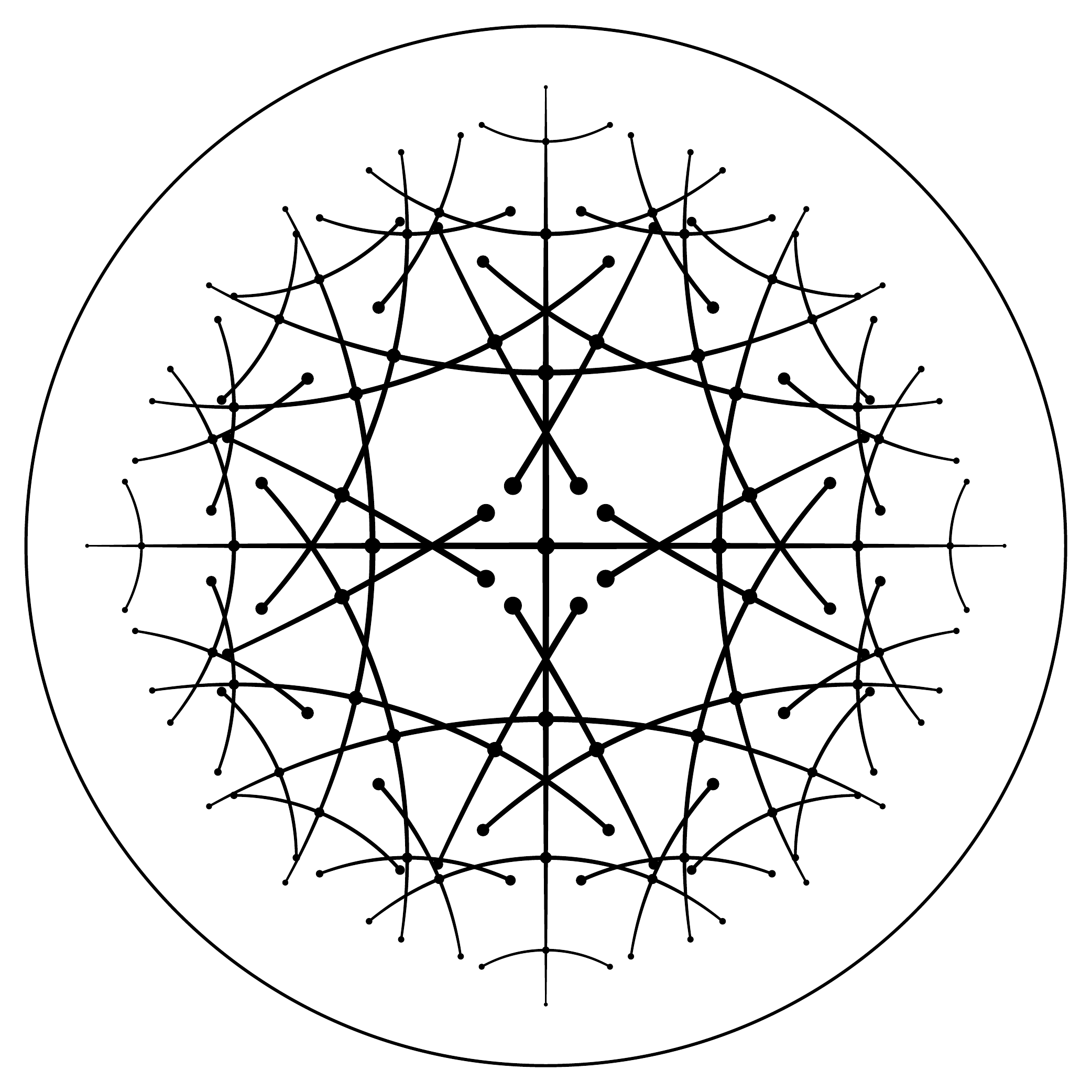}
         \caption{$\mathcal{M}_4$}
     \end{subfigure}
     \hfill
     \begin{subfigure}[b]{0.3\textwidth}
         \centering
         \includegraphics[scale=0.17]{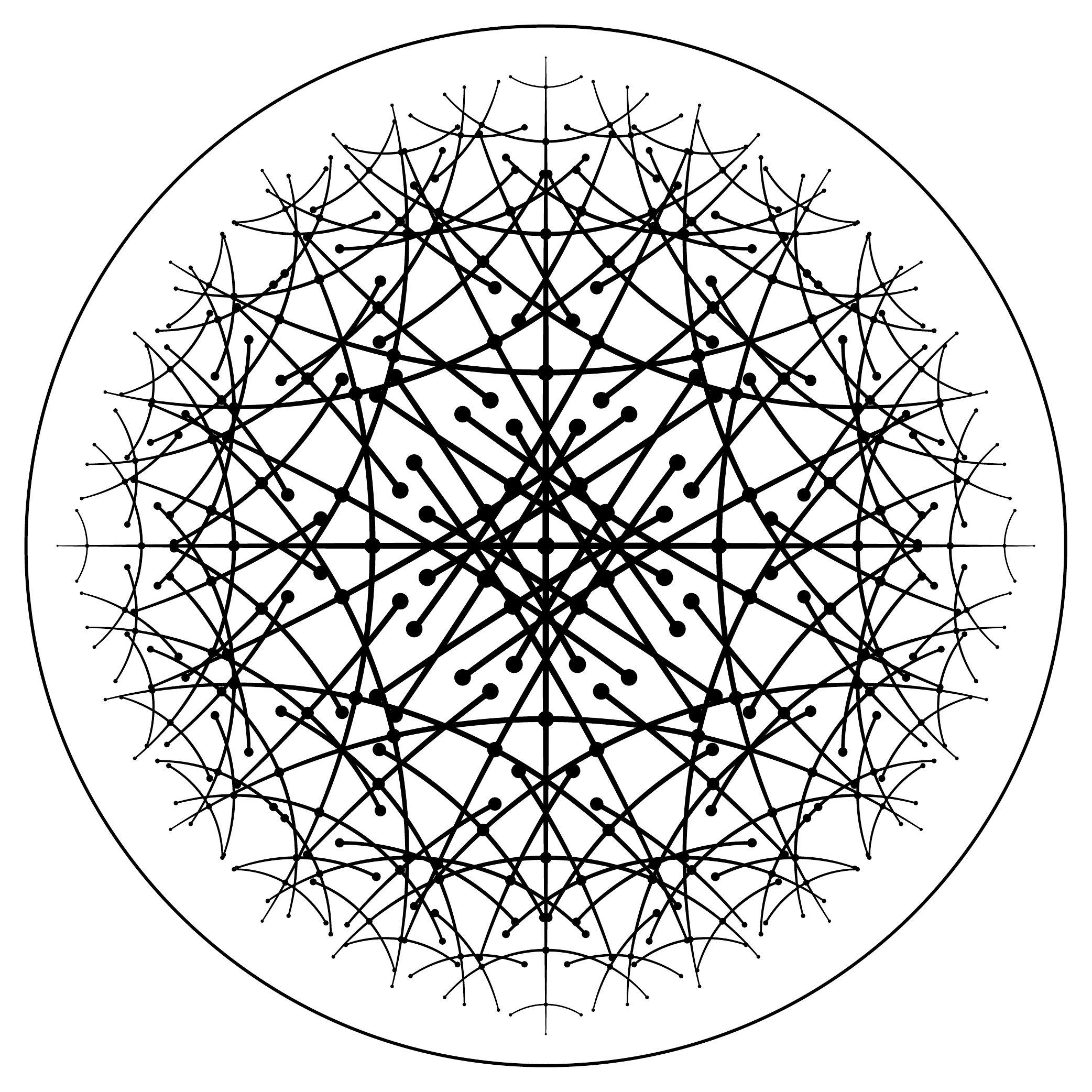}
         \caption{$\mathcal{M}_5$}
     \end{subfigure}
     \par\bigskip
     \begin{subfigure}[b]{0.3\textwidth}
         \centering
         \includegraphics[scale=0.17]{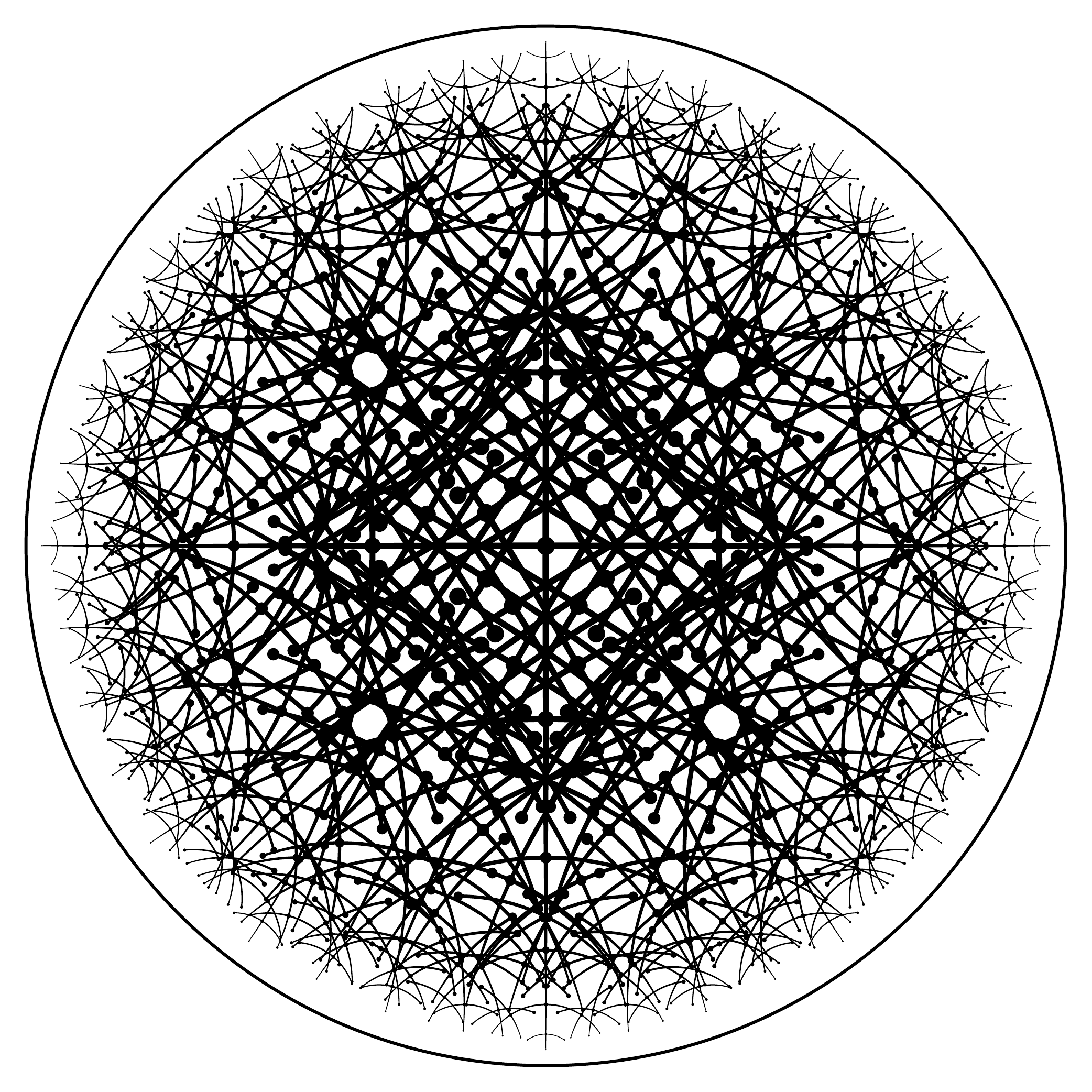}
         \caption{$\mathcal{M}_6$}
     \end{subfigure}
     \hfill
     \begin{subfigure}[b]{0.3\textwidth}
         \centering
         \includegraphics[scale=0.17]{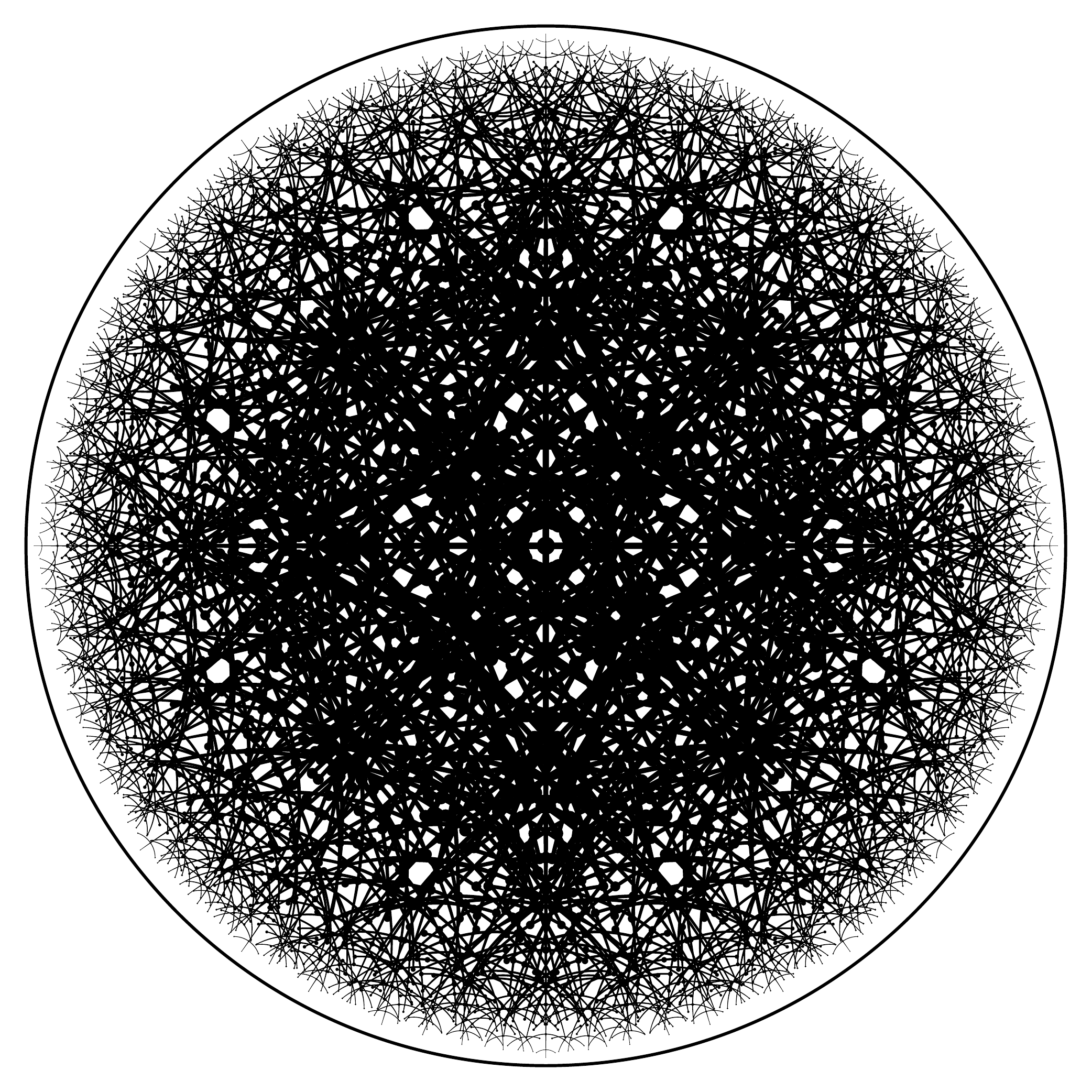}
         \caption{$\mathcal{M}_7$}
     \end{subfigure}
     \hfill
     \begin{subfigure}[b]{0.3\textwidth}
         \centering
         \includegraphics[scale=0.34]{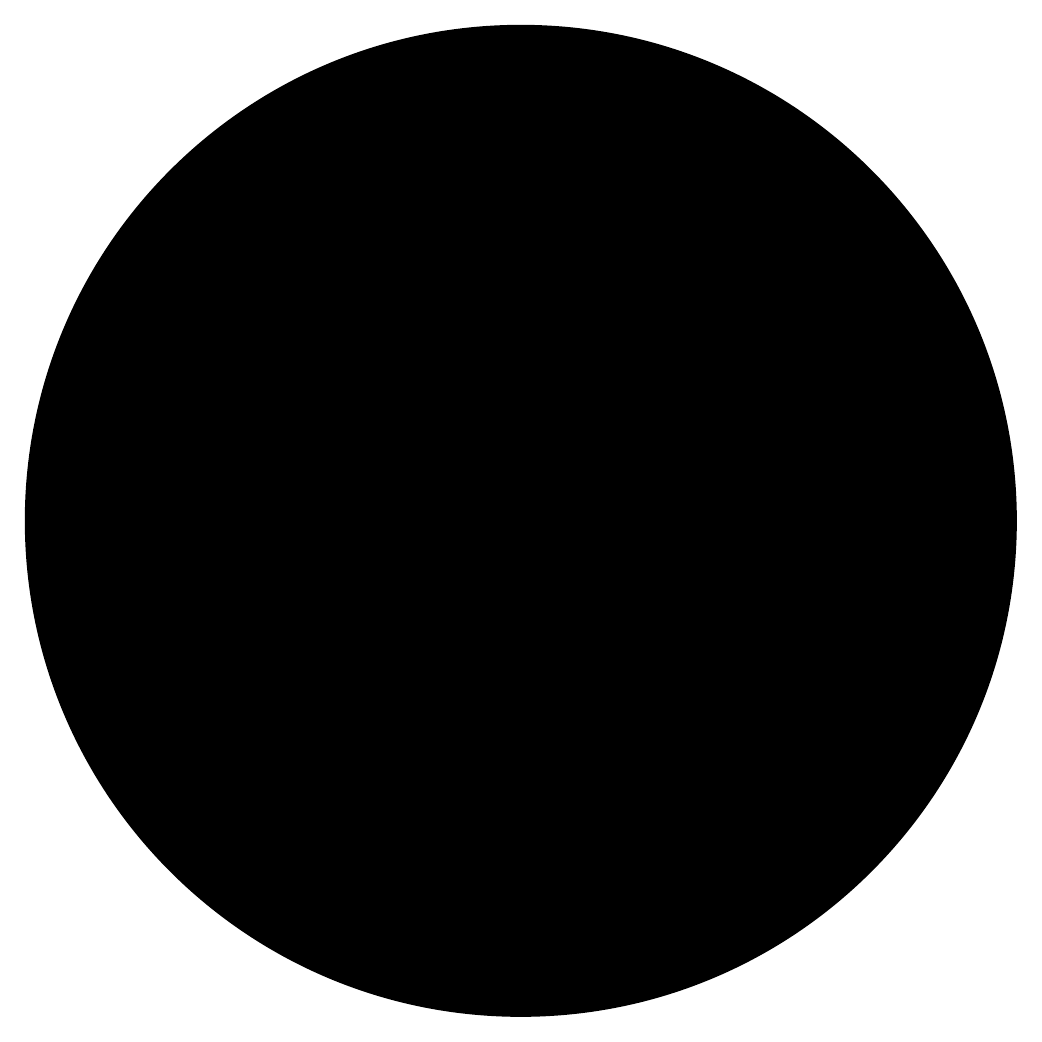}
         \caption{$\mathcal{M}_\infty$}
     \end{subfigure}
        \caption{
        Depicted are nine Poincaré disk representations of the hyperbolic space of the Lorentz group. 
        Each dot represents one frame-grid.
        If two dots are connected by one line, they are interlaced, as previously described, corresponding to a Lorentz boost with a rapidity $w = \ln(2)$ . 
        The first disk corresponds to \mbox{Figure \ref{fig:2F}} of the previous section.
        The other eight disks are 2D cross sections of 3D Poincaré ball models corresponding to graphs $\mathcal{M}_n$ that can be seen to be constructed through recursive interlacing of further frame-grids.
        As $n$ goes towards infinity, the hyperbolic space is filled completely, as depicted in the last disk according to \mbox{Theorem \ref{fullLorentz}} .\\ 
        We provide the program to generate these images \mbox{at \cite{leuenberger2021}.}\\
        }
        \label{fig:nine disks}
\end{figure}
\end{samepage}

\subsubsection{Emergent Isotropy}
The Wigner rotation is important for the GRIDS, since its infinitely repeated occurrence yields all possible orientations combined with all possible rapidities, and thus yields isotropy without requiring any additional rules, 
i.e.: The rules of \mbox{Section \ref{Euclid}} become obsolete here.
We show this claim to be true in the proof of \mbox{Theorem \ref{fullLorentz}}. 
In preparation for it, we firstly prove the following lemma.
\\
\begin{lemma}\label{circle1}
In the hyperbolic plane, let us perform a sequence of steps, each covering the same distance $w$ ; between the steps we change our direction by the angle $\phi$. 
Then the resulting set of visited points will be uniformly distributed on a circle, if $w$ and $\phi$ meet the following condition:
\[4\ (\cosh(\frac{w}{2})\sin(\frac{\pi-\phi}{2}))^2 \ \in \ [0,4] \cap \mathbb{Q} \setminus \mathbb{Z} \]
\end{lemma}

\begin{proof}
Starting at a point $B_0$, perform one of the described steps to visit another point $B_1$. 
Let the center of a circle, on which all visited points are located, be called $A$.
Then the center $A$ must lie on the perpendicular bisector line of the line segment $B_0 B_1$. 
Let the midpoint of the line segment $B_0 B_1$ be called $C$. 
Continuing from $B_1$, perform the next step to visit the another point $B_2$. 
The center $A$ then must also lie on the angle bisector line of the angle $\measuredangle B_0 B_1 B_2$. 
Note that the hyperbolic triangle $\bigtriangleup A B_1 C$ is a right triangle, since $\measuredangle A C B_1 = \frac{\pi}{2}$. 
Let $\alpha := \measuredangle B_1 A C$ , $\beta := \measuredangle C B_1 A$ , and $a := |B
_1 C|$ , as is usual in trigonometry.\\
Due to the bisections, we get $a = \frac{|B_0 B_1|}{2} = \frac{w}{2}$ as well as $\beta = \frac{\measuredangle B_0 B_1 B_2}{2} = \frac{\pi-\phi}{2}$ .\\
The following equation applies to right hyperbolic triangles such as $\bigtriangleup A B_1 C$ :\\
$\cos(\alpha) = \cosh(a) \sin(\beta)$ , from which follows:
$\alpha = \arccos( \cosh(a) \sin(\beta) )$ .\\
By taking further steps, we visit the set of points $B = \{B_0,B_1,B_2,B_3,..\}$.
Due to the required regularities, each step moves us around the center $A$ by the same angle that is equal to $2\alpha$, due to the bisections, 
i.e.: $\frac{\measuredangle B_{k-1} A B_k}{2} = 2\alpha$ , $\forall k \in \mathbb{N}$ .\\
Therefore, if $\frac{2\alpha}{2\pi} = \frac{\alpha}{\pi}$ is an irrational number, we will never revisit the starting point exactly, no matter how many laps we completed on the circle. 
Furthermore, according to the \textit{equidistribution theorem}\cite{weyl1916gleichverteilung}, if $\frac{\alpha}{\pi}$ is an irrational number, then the points will be uniformly distributed on the circle.\\
Therefore, in order to prove our lemma, we only need to show that the condition \
$4(\cosh(\frac{w}{2})\sin(\frac{\pi-\phi}{2}))^2 \in [0,4] \cap \mathbb{Q} \setminus \mathbb{Z}$ \ implies the irrationality of $\frac{\alpha}{\pi}$ , 
i.e.: $\frac{\alpha}{\pi}\not\in\mathbb{Q}$.\\
Recall that $a = \frac{w}{2}$ and $\beta = \frac{\pi-\phi}{2}$ . 
We can use these two equations to rewrite the condition as follows:\
$4(\cosh(a)\sin(\beta))^2 \in [0,4] \cap \mathbb{Q} \setminus \mathbb{Z}$ .\\
We then substitute $(\cosh(a)\sin(\beta))^2$ with a variable $r$ : \ \ $4r \in [0,4] \cap \mathbb{Q} \setminus \mathbb{Z}$ .\\
The theorem about the arccosine function of Varona \cite{varona2006rational} then implies that the number $\frac{1}{\pi}\arccos(\sqrt{r})$ cannot be rational, 
i.e.: $\frac{1}{\pi}\arccos(\sqrt{r}) \not\in \mathbb{Q}$.\\
We then unsubstitute $r$ and get:
$\frac{1}{\pi}\arccos(\cosh(a)\sin(\beta)) \not\in \mathbb{Q}$ .\\
Recall that $\alpha = \arccos( \cosh(a) \sin(\beta) )$ ;
and thus it follows that $\frac{\alpha}{\pi}\not\in\mathbb{Q}$.
\\
\textbf{q.e.d.}
\end{proof}
\ \\
\begin{corollary}\label{circle2}
If the sequence of steps is infinite, then the set of visited points will form a continuous circle.\\
\\
\end{corollary}

\begin{theorem}{\ \textbf{$\mathcal{M}_\infty$ fills the Lorentz group:}}\label{fullLorentz}
\\
Let $\mathcal{M}_\infty$ be $\mathcal{M}_n$ for $ n \rightarrow \infty$. \
The set of Lorentz transformations corresponding to all frame-grids of $\mathcal{M}_\infty$ is equal to the continuous set of all Lorentz transformations, i.e.: The entire Lorentz group.
\\
\end{theorem}

\begin{proof}
We firstly need to show that, for each point in the 3D hyperbolic space, it is possible to reach the point through some sequence of steps corresponding to the Lorentz transformations corresponding to some sequence of interlaced frame-grids within $\mathcal{M}_\infty$, when starting from the central frame-grid, that is the center point in the the disk models of \mbox{Figure \ref{fig:nine disks}} . 
Recall that the rapidity $w$ of the individual Lorentz boosts between frame-grids in $\mathcal{M}_n$ is equal to $\ln(2)$. 
In $\mathcal{M}_n$, the angle by which the direction changes after each boost, is either zero or $\frac{\pi}{2}$. 
Let us therefore apply \mbox{Lemma \ref{circle1}} with $w=\ln(2)$ and $\phi=\frac{\pi}{2}$.
We need to verify whether the condition of \mbox{Lemma \ref{circle1}} holds for these values of $w$ and $\phi$ :
\\
$ 4\ (\cosh(\frac{w}{2})\sin(\frac{\pi-\phi}{2}))^2 
\ = \ 4\ (\cosh(\frac{\ln(2)}{2})\sin(\frac{\pi}{4}))^2 
\ = \ 4\ (\frac{3}{2\sqrt{2}}\cdot \frac{1}{\sqrt{2}})^2 
\ = \ \frac{9}{4}$ 
\\
Now since \ $\frac{9}{4}
\in \ [0,4] \cap \mathbb{Q} \setminus \mathbb{Z} $ \ is true, the condition is fulfilled.\\
Due to the infinities, we can furthermore apply \mbox{Corollary \ref{circle2}}. 
It implies that the set of Lorentz transformations corresponding to all frame-grids of $\mathcal{M}_\infty$ forms a shape that contains many continuous circles in the 3D hyperbolic space of the Lorentz group. 
As a side-note, it therefore forms continuous helices in the Lorentz group itself, which is 6-dimensional, due to the three additional degrees of freedom for rotations.  
In the 3D hyperbolic space, we can thus choose a sequence of perpendicular steps of length $w$ to move arbitrarily close to any point on a continuous circle, thus also allowing us to move by arbitrarily small distances from the origin. 
At any point of this circle we can then choose to keep moving on a different circle, that is perpendicular to the previous circle.
We can then change to further circles an arbitrary number of times, resulting in a path that can be thought of as a composition or concatenation of perpendicular circular arcs. 
Such a path has enough degrees of freedom to reach any point of the 3D hyperbolic space.
Furthermore, such a path can reach any point at an arbitrary orientation, thus the entire Lorentz group, which is six-dimensional, is continuously filled by the Lorentz transformations of the frame-grids of $\mathcal{M}_\infty$.
\\
\textbf{q.e.d.}
\end{proof}
$\ $

\subsubsection{Speed of Light}
Recall \mbox{Theorem \ref{euclid},} where the geodesic distance on a graph approximated the Euclidean distance.
The following theorem, \mbox{Theorem \ref{luminal}}, is analogous to this, in that the speed of the quickest path along light-like edges of $\mathcal{M}_\infty$ approximates a constant speed equal to one, through the 3D Euclidean space.
Let these speeds be called the \textit{speed of light} for simplicity.
The following theorem then says that the longer the path is that the light travels, the more accurately its speed will approximate a constant in any direction, i.e.: isotropically.
Some concrete examples of calculations are provided below the proof.
\\
\begin{samepage}
\begin{theorem}\label{luminal}
{\textbf{Accuracy of the Speed of Light on $\mathcal{M}_\infty$}
}\\
For all frame-grids $F_0$ of $\mathcal{M}_\infty$:\\
Given two randomly selected spatial locations $\Vec{q},\Vec{u}\in\mathbb{Z}^3$ on $F_0$:\\
Let \ $s := |\Vec{q}-\Vec{u}|$ . 
Consider the following path: 
Starting at a vertex of $F_0$, located at $\Vec{q}$, at a randomly selected departure time, move along \mbox{light-like edges} of $\mathcal{M}_\infty$ to arrive as quickly as possible at another vertex of $F_0$, located at $\Vec{u}$, at the resulting arrival time.
Let the integer $\Delta t$ then be the difference between the departure and the arrival time in the frame of reference of $F_0$.
The following inequality then expresses how little the velocity \ $\frac{s}{\Delta t}$ will deviate from one : 
\[\text{expected error of the speed of light}
\ \ := \ \ \EX(|1 - \frac{s}{\Delta t}|)
\ \ < \ \
150\cdot\frac{ \log_2(s) +6 }{s}\]
\[= \ \ \mathcal{O}\left(\frac{\log s}{s}\right)
\]
\end{theorem}
{\ }\end{samepage}
\begin{proof}
To prove this formula, we show the existence of a sufficiently quick path. 
Our path starts on frame-grid $F_0$ at location $\Vec{q}$, then usually leads through multiple other frame-grids of $\mathcal{M}_\infty$ before returning back to the same frame-grid $F_0$ but at a different location that is $\Vec{u}$.
This path is a alternation of several long straight paths on single frame-grids with several shorter non-straight paths that lead through multiple frame-grids. 
It is therefore analogous to the path that we described back in \mbox{Theorem \ref{euclid}} , but consists of seven instead of five parts, due to the increased number of dimensions.

In what follows, for simplicity, we will restrict the set of frame-grids that the path is allowed to lead through.
Recall from the proof of \mbox{Theorem \ref{fullLorentz}} that the condition of \mbox{Lemma \ref{circle1}} is fulfilled for $w=\ln(2)$ and $\phi=\frac{\pi}{2}$.
Hence, there are sequences of interlaced frame-grids of $\mathcal{M}_\infty$, such that their corresponding sequences of points within the 3D hyperbolic space are uniformly distributed on circles. 
Also note that each such point is part of twelve such circles that are located on three perpendicular planes; obviously, these twelve circles are the pairwise intersections of eight spheres.
We are concerned with only one of these spheres.
Analogously to how a the sequence resulted in a uniform distribution over a circle, a binary tree can result in a uniform distribution over an aforementioned sphere.
For simplicity, we restrict the set of frame-grids that our path is allowed to lead through, to such a spherical binary tree, where each node is a frame-grid.
One of these frame-grids must be $F_0$.
Each edge in the binary tree indicates that the two frame-grids are interlaced with each other.
In the 3D hyperbolic space, the points corresponding to these frame-grids are uniformly distributed across a sphere.
Let this binary tree be an \textit{unrooted binary tree} containing $2^n$ nodes, such that the number of steps required to move between any two nodes is at most $n$.
We will set the parameter $n$ later in this proof.
Note that not only are these points uniformly distributed across that sphere, but also the orientations of these frame-grids are uniformly distributed.

Recall that the relative rapidity between interlaced frame-grids was $w = \ln(2)$ and in the binary tree, each non-leaf node has three neighbouring nodes. 
This translates to a point on the sphere, let it be called $A$, being at a distance of $\ln(2)$ from three other points on the sphere in three perpendicular directions.
We now use this fact in order to calculate the radius of this sphere within the hyperbolic space.
We use a right hyperbolic triangle, where the hypotenuse is the line segment from the sphere's center to the point $A$, which is of length $r$, i.e.: The radius of the sphere. 
One of the catheti, $b$, is a half of the line segment from $A$ to one of the three aforementioned points.
This cathetus $b$ is thus of length $\frac{\ln(2)}{2}$.
Due to the symmetry between the three aforementioned points located in perpendicular directions, the angle $\alpha$ at $A$ must be equal to the angle between the diagonal of a cube and one of its edges.
We therefore obtain $\alpha = \arctan(\sqrt{2})$, which we can use to calculate the length $r$ of the hypotenuse as follows: \
$r \ = \ \artanh(\frac{\tanh(b)}{\cos(\alpha)})
\ = \ \artanh(\frac{\tanh(\frac{\ln(2)}{2})}{\cos(\arctan(\sqrt{2}))})
\ = \ \artanh(\sqrt{3^{-1}})$ \\
We use this radius later for the calculation of time delays.

Let $F_1$, $F_2$, and $F_3$ be frame-grids that should be thought of having orientations that are approximately aligned with the straight line from $\Vec{q}$ to $\Vec{u}$, 
i.e.: Some of the most well aligned out of the set of $2^n$ frame-grids.
To take the most direct path, the largest portion of the distance of the path is covered within these three approximately aligned frame-grids.
There are three of them because three vectors can be linearly composed to reach any point within a 3D volume.
The entire path can be split into the following sequence of seven paths:
\\
(1) Quickest path from $F_0$ at location $\Vec{q}$ to $F_1$.\\
(2) Straight light-like path within $F_1$ on line $l_1$.\\
(3) Quickest path from $l_1$ in $F_1$ to $l_2$ in $F_2$.\\
(4) Straight light-like path within $F_2$ on line $l_2$.\\
(5) Quickest path from $l_2$ in $F_2$ to $l_3$ in $F_3$.\\
(6) Straight light-like path within $F_3$ on line $l_3$.\\
(7) Quickest path from $F_3$ to $F_0$ at location $\Vec{u}$.\\
While the straight paths (2), (4), and (6) are located within single frame-grids, the other paths (1), (3), (5), and (7) are not straight and lead through multiple different frame-grids in as little time as possible, thus only covering little distance when compared to the straight paths. 

We now calculate the worst case time delays caused by the four non-straight paths (1), (3), (5), and (7).
The relative rapidity between frame-grids in the binary tree must be smaller or equal to the largest distance between their positions in the hyperbolic space, which is equal to the diameter $2r$ of the sphere. 
The distance travelled along a single light-like edge of any frame-grid of the tree would then have to be \mbox{at most $e^{2r}$} and \mbox{at least $e^{-2r}$} in the frame of reference \mbox{of $F_0$}.
Also the elapsed time would then accordingly have to be \mbox{between $e^{2r}$ and $e^{-2r}$.}
In the absolute worst case the path would lead across a light-like edge that leads in the opposite direction of the path's final destination, thus causing a time delay of $2 e^{2r}$ . 
For any two interlaced frame-grids, the maximal number of steps along light-like edges required in order to move from any vertex of the first frame-grid to a vertex shared with the second frame-grid is just \textit{two}, as can be derived from the rules in \mbox{Subsection \ref{4Rules}} .
This yields a worst case time delay of \mbox{less than $4 e^{2r}$} per move from one frame-grid to the next one. 
Recall that the number of steps required to move
between any two frame-grids that are nodes of the binary tree is at most n.
This results in a time-delay of \mbox{less than $4n e^{2r}$} per non-straight path, of which there are four, hence resulting in \mbox{less than $16n e^{2r}$} .
By inserting the previously calculated value for the radius $r$ we then get: 
\\
$16n\cdot e^{2r} \ = \ 16n\cdot \exp(2 \artanh(\sqrt{3^{-1}})) \ = \ 16n \cdot (2+\sqrt{3}) \ < \ 60 n $ .

We now proceed to the other time delays that are caused by the small deviations of the orientations of the three straight paths (2), (4), and (6) from the direction of the 'direct line' from $\Vec{q}$ to $\Vec{u}$.
These deviations of directions are three small angles.
Let $\varepsilon$ be a small value that is larger than each of these three angles. 
The paths would thus at most be elongated by the factor  $\cos(\varepsilon)^{-1}$ .
Recall that $\Delta t$ is the total time taken and $s := |\Vec{q}-\Vec{u}|$ .
By adding all time-delays, we obtain the bound shown in the following inequality:
\\
$\Delta t \ \
< \ \ s\cdot \cos(\varepsilon)^{-1} + 60n$
\\
$\Delta t \ \
< \ \ s\cdot (1+\varepsilon^2) + 60n$
\hspace*{\fill} ( since \ $\cos(\varepsilon)^{-1} \leq 
1+\varepsilon^2$ \ \ for \ $0 \leq \varepsilon \leq 1$ )
\\
$\frac{\Delta t}{s} \ \ 
< \ \ 1 + \varepsilon^2 + 60\frac{n}{s}$
\\
$\frac{\Delta t}{s} -1 \ \ 
< \ \ \varepsilon^2 + 60\frac{n}{s}$
\\
$1 - \frac{s}{\Delta t} \ \ 
< \ \ \varepsilon^2 + 60\frac{n}{s}$
\hspace*{\fill} ( since \ $1 - \frac{s}{\Delta t} 
\leq \frac{\Delta t}{s} -1$ 
\ \ for \ $0 < s \leq \Delta t$ )
\\
On the left side of this last formula is the error of the fastest average speed, i.e.: the speed of light.
We let the unknown values $\varepsilon$ and $n$ disappear from the right side in what follows.
We firstly rewrite the previous formula in terms of expectation values:
\\
$\EX(1 - \frac{s}{\Delta t}) \ \ 
< \ \ \EX(\varepsilon^2 + 60\frac{n}{s}) \ \
= \ \ \EX(\varepsilon)^2 + \Var(\varepsilon) + 60\frac{n}{s}$
\\
We now need to calculate bounds for the expectation value $\EX(\varepsilon)$ and the variance $\Var(\varepsilon)$ of a minimal $\varepsilon$ depending on $n$. 
The larger the parameter $n$ is, the more frame-grids with different orientations there are available to choose from, which allows for a smaller $\varepsilon$ to exist, which implies a smaller time-delay on the straight paths.
Each frame-grid provides six possible directions for straight light-like paths.
Note that two interlaced frame-grids share directions with each other, such that together they provide only ten rather than twelve such directions.
To count the total number of directions provided by the binary tree, we can therefore count six per node minus two per connection between nodes, i.e.:
\mbox{$m = 6 \cdot 2^n - 2 \cdot (2^n-1) > 2^{n+2}$} , where $m$ is the total number of possible directions of straight light-like paths within frame-grids that are the nodes of the binary tree.
Let $\varphi$ be an angle that is the 
\mbox{\textit{great-circle distance}}
between two points that are sampled from a uniform distribution on a unit sphere.
The probability $p$ that $\varphi$ will be less or equal to $\varepsilon$ will then obviously be \mbox{equal to $\frac{1}{2}(1-\cos(\varepsilon))$ ,} \ i.e.: \
\mbox{$p \ := \ P(\varphi \in [0,\varepsilon] ) 
\ = \ \frac{1}{2}(1-\cos(\varepsilon))$} .
\\
Then \ \mbox{$(1-(1-p)^m)\cdot(1-(1-p)^{m-1})\cdot(1-(1-p)^{m-2})$} \ is the probability that, within a \mbox{great-circle distance} of $\varepsilon$ of a given point, there will be at least three points out of a set of $m$ points sampled from the uniform distribution over the whole unit sphere.
\\
The three directions in 3D space of the paths (2), (4), and (6), when mapped to points on the unit sphere, are the corners of a small spherical triangle that needs to encase the point that is the direction from $\Vec{q}$ to $\Vec{u}$, otherwise the path would miss its target. 
For three randomly sampled corners there is obviously a probability of $\frac{1}{2}$ that this triangle will encase that point. 
To obtain the probability $H(\varepsilon)$ that the whole path exists we therefore multiply the last exponent in the aforementioned probability \mbox{with $\frac{1}{2}$ :}
\\
$H(\varepsilon) \ \ := \ \ (1-(1-p)^m)\cdot(1-(1-p)^{m-1})\cdot(1-(1-p)^{\frac{1}{2}\cdot(m-2)})$
\\
$\geq \ \ (1-(1-p)^{\frac{1}{2}\cdot(m-2)})^3 $
\\
$= \ \ (1-(1-\frac{1}{2}(1-\cos(\varepsilon)))^{\frac{1}{2}\cdot(m-2)})^3 $
\hspace*{\fill} 
(since \ $p = \frac{1}{2}(1-\cos(x))$)
\\
$\geq \ \ (1-(1-\frac{\varepsilon^2}{5})^{\frac{1}{2}\cdot(m-2)})^3 $
\hspace*{\fill} 
(since \ $\frac{1}{2}(1-\cos(\varepsilon))  \geq  \frac{\varepsilon^2}{5}$ \ 
for $0 \leq \varepsilon \leq 1$)
\\
$= \ \ \left(1-{\sqrt{1-\frac{\varepsilon^2}{5}}}^{\ m-2}\right)^3$
$ =: \ G(\varepsilon)$
\\
The smaller function $G$ was introduced in order to simplify the terms. \\
For the minimized $\varepsilon$, the derivatives then give us the\\ \mbox{probability density functions (pdf) $h$ and $g$ :} \\
$h(\varepsilon) \ := \ \frac{\partial}{\partial \varepsilon} 
H(\varepsilon)
$\\
$g(\varepsilon) \ := \ \frac{\partial}{\partial \varepsilon} 
G(\varepsilon)
$\\
As previously shown $G$ is smaller than $H$ for $\varepsilon \leq 1$ . 
Larger values for $\varepsilon$ are not considered since the minimal $\varepsilon$ will be much smaller than $1$ for large $m$.
Because $G(\varepsilon)$ increases slower than $H(\varepsilon)$, it follows that the pdf $g$ is more spread out than the pdf $h$ and therefore its expectation value is larger as well as its variance is larger, i.e.:\\
$\EX_h(\varepsilon) \ < \ \EX_g(\varepsilon)$ \ and \ $\Var_h(\varepsilon) \ < \ \Var_g(\varepsilon) $
\\
We now derive an upper bound for the expectation value:\\
$\EX(\varepsilon) 
\ = \ \EX_h(\varepsilon) 
\ < \ \EX_g(\varepsilon) 
\ = \ 
\int_0^{1} \
\varepsilon\cdot
g(\varepsilon) \ d\varepsilon$
$\ = \ 
\int_0^{1} \
\varepsilon\cdot
\frac{\partial}{\partial \varepsilon} 
G(\varepsilon) \ d\varepsilon $
\\
$\ = \ 
\int_0^{1} \
\varepsilon\cdot
\frac{\partial}{\partial \varepsilon} 
\left(1-{\sqrt{1-\frac{\varepsilon^2}{5}}}^{\ m-2}\right)^3
\ d\varepsilon$
\\
$ \ = \ 
\int_0^{1} \ 
\varepsilon \cdot
\frac{3 \varepsilon}{5}(m-2) \
\sqrt{1-\frac{\varepsilon^2}{5}}^{\ m-4}
\cdot
\left(\sqrt{1-\frac{\varepsilon^{2}}{5}}^{\ m-2} -1\right)^{2}
\ d\varepsilon$
\\
$ \ < \ 
\frac{3m}{5}\cdot \
\int_0^{1} \ 
\varepsilon^2 \cdot
\sqrt{1-\frac{\varepsilon^2}{5}}^{\ m-4}
\cdot
\left(\sqrt{1-\frac{\varepsilon^{2}}{5}}^{\ m-2} -1\right)^{2}
\ d\varepsilon$
\\
$ \ < \ 
\frac{3m}{5}\cdot \
\int_0^{1} \ 
\varepsilon^2 \cdot
\sqrt{1-\frac{\varepsilon^2}{5}}^{\ m-4}
\ d\varepsilon$
\\
$ \ < \
\frac{3m}{5}\cdot \
\int_0^{1} \ 
\varepsilon^2 \cdot
\ \left(1-\frac{\varepsilon^2}{10}\right)^{m-4}
d\varepsilon$
\hspace*{\fill} 
( \ For any $ \kappa \in [0,1] $ :\ )
\\
$ \ = \ 
\frac{3m}{5}\cdot \ 
( \
\int_0^\kappa \ 
\varepsilon^2 \cdot
\left(1-\frac{\varepsilon^2}{10}\right)^{m-4}
d\varepsilon
\ \ + \
\int_\kappa^1 \ 
\varepsilon^2 \cdot
\left(1-\frac{\varepsilon^2}{10}\right)^{m-4}
d\varepsilon
\ )$
\\
$ \ < \ 
\frac{3m}{5}\cdot \ 
( \ \kappa^2 \cdot
\int_0^\kappa \ 
1^{m-4}
\ d\varepsilon
\ \ + \
\int_\kappa^1 \ 
1 \cdot
\left(1-\frac{\varepsilon^2}{10}\right)^{m-4}
d\varepsilon
\ )$
\\
$ \ < \ 
\frac{3m}{5}\cdot \ 
\left( \ \kappa^3
\ \ + \
\left(1-\frac{\kappa^2}{10}\right)^{m-4}
\ \right)$
\hspace*{\fill} 
(\ Let $ \kappa =\ 4\cdot \sqrt[5]{m^{-2}}$ :\ )
\\
$ \ = \ 
\frac{3m}{5}\cdot \ 
\left( \ 4^3\cdot {\sqrt[5]{m^{-6}}}
\ \ + \
\left( 1 -\frac{4^2}{10} \cdot\sqrt[5]{m^{-4}} \right)^{m-4}
\ \right)$
\\
$ \ < \ m\cdot 
4^3\cdot {\sqrt[5]{m^{-6}}}$
$\ \ = \ \ \frac{64}{\sqrt[5]{m}}$
\\
Next we derive an upper bound for the variance:
\\
$\Var(\varepsilon) 
\ = \ \Var_f(\varepsilon) 
\ < \ \Var_g(\varepsilon) $
$\ = \ 
\int_0^{1} \
\varepsilon^2\cdot
g(\varepsilon) \ d\varepsilon$
$\ = \ 
\int_0^{1} \
\varepsilon^2\cdot
\frac{\partial}{\partial \varepsilon} 
G(\varepsilon) \ d\varepsilon $
\\
$ \ < \ 
\frac{3m}{5}\cdot \ 
\left( \ \kappa^4
\ \ + \
\left(1-\frac{\kappa^2}{10}\right)^{m-4}
\ \right)$
\hspace*{\fill} ( \ By analogous steps to earlier. \ )
\\
$ \ = \ 
\frac{3m}{5}\cdot \ 
\left( \ 3^4 \cdot \sqrt[5]{m^{-7}}
\ \ + \
\left(1-\frac{3^2}{10}\cdot 
\sqrt[20]{m^{-14}}
\right)^{m-4}
\ \right)$
\hspace*{\fill} 
( \ with $ \kappa =\ 3\cdot \sqrt[20]{m^{-7}}$ \ )
\\
$ \ < \ m \cdot 
3^4 \cdot \sqrt[5]{m^{-7}}$
$ \ \ = \ \ \frac{81}{\sqrt[5]{m^2}}$
\\
So to summarize:\\ $\EX(\varepsilon) 
\ < \ \frac{64}{\sqrt[5]{m}}$
 \ and \ 
$\Var(\varepsilon) 
\ < \ \frac{81}{\sqrt[5]{m^2}}$
\\
We can now insert these two values into our earlier formula and afterwards set the parameter $n$ depending on $s$:
\\
$\EX(1 - \frac{s}{\Delta t}) \ \ 
< \ \ \EX(\varepsilon^2 + 60\frac{n}{s}) \ \
= \ \ \EX(\varepsilon)^2 + \Var(\varepsilon) + 60\frac{n}{s}$
\\
$= \ \ (\frac{64}{\sqrt[5]{m}})^2
+ \frac{81}{\sqrt[5]{m^2}}
+ 60\frac{n}{s}$
$\ \ = \ \ (\frac{64}{\sqrt[5]{{2^n}}})^2
+ \frac{81}{\sqrt[5]{{(2^n)^{2}}}}
+ 60\frac{n}{s}$
\\
$\ \ < \ \ 2^{(12-\frac{2}{5}n)}
+ 60\frac{n}{s}$
\hspace*{\fill} 
( \ Let $n \ := \ 
\lceil \frac{5}{2}\cdot \log_2(s) +10 \rceil$ \ :  )
\\
$\ \ = \ \ 2^{(12-\frac{2}{5}
\lceil \frac{5}{2}\cdot \log_2(s) +10 \rceil)}
+ 60\frac{\lceil \frac{5}{2}\cdot \log_2(s) +10 \rceil}{s}$
\\
$\ \ < \ \ \frac{256}{s}
+ 60\frac{\lceil \frac{5}{2}\cdot \log_2(s) +10 \rceil}{s}$
\\
$\ \ < \ \ 
150\cdot\frac{ \log_2(s) +6 }{s}$
\\
\textbf{q.e.d.}
\end{proof}
{\ }
\\
\textbf{Examples for Theorem \ref{luminal}}
\\
Here we provide examples where we use \mbox{Theorem \ref{luminal}} to calculate upper bounds on the deviation of the speed of light travelling one metre, one parsec, and one \r{a}ngström in any direction.
As the Planck length $\ell_P$ is often speculated to be the fundamental length, here we set the light-like edges to cover a distance equal to one $\ell_P$ within their respective frame of reference. 
To convert to SI units, we multiply with the speed of light constant 
$c\ =\ 299\ 792\ 458\ \frac{metre}{second}$ .
\\
For light travelling one metre we get:\ \
$s\ =\ \frac{metre}{\ell_P} 
\ =\ 6.25..\cdot10^{34}$\\
We then insert this $s$ into the formula of the theorem:
\\
deviation
$\ \ < \ \ 
c\cdot150\cdot\frac{ \log_2(s) +6 }{s}$
$\ \ = \ \ 
8.84..\cdot10^{-23} \frac{metre}{second}$\\
This indicates that the speed of light deviates only by an extremely small speed, which means high accuracy for one metre.
\\
For light travelling one parsec we get:\ \
$s \ = \ \frac{parsec}{\ell_P} 
\ = \ 1.93..\cdot10^{51}$\\
deviation
$\ \ < \ \ 
c\cdot150\cdot\frac{ \log_2(s) +6 }{s}$
$\ \ = \ \ 
4.15..\cdot10^{-39} \frac{metre}{second}$\\
We can see that for this astronomical distance the speed of light becomes 15 orders of magnitude more accurate, theoretically. 
But more interestingly; is it still accurate at microscopic distances, such as one \r{a}ngström? \\
For light travelling one \r{a}ngström we get:\ \
$s \ = \ \frac{\text{\r{a}ngström}}{\ell_P} 
\ = \ 6.25..\cdot10^{24}$
\\
deviation
$\ \ < \ \ 
c\cdot150\cdot\frac{ \log_2(s) +6 }{s}$
$\ \ = \ \ 
6.42..\cdot10^{-13} \frac{metre}{second}$
\\
In conclusion, even at a distance as short as one \r{a}ngström, the speed of light is still so accurate that the speed by which it can deviate is still around a hundred times slower than the speed of fingernail growth or the speed of continental drift.
It should also be noted that with more work our bound could be optimized even further.

\mbox{Theorem \ref{luminal}} also provided the general error term $ \mathcal{O}\left(\frac{\log s}{s}\right) $ that applies to not only the speed of light on $\mathcal{M}_\infty$, but to the speed on light on other GRIDS as well. 
A similar term will reappear concerning the proper time interval in the following subsection.
\\

\subsubsection{Proper Time Interval}
In special relativity, the \textit{proper time interval} $\Delta \tau$ of a geodesic between between two events is given by the formula $\Delta \tau = \sqrt{(\Delta t)^2 - (\frac{s}{c})^2}$ , where $s$ is the spatial distance between the two events and $\Delta t$ is the time difference. 
While $\Delta t$ and $s$ differ depending on the inertial frame of reference, $\Delta \tau$ does not, i.e.: 
$\Delta \tau$ is Lorentz invariant. 
The difference between $\Delta t$ and $\Delta \tau$ is called \textit{time dilation}. 
$\Delta \tau$ can also be written in terms of a Lorentz factor $\gamma$ as follows:
\mbox{$\Delta \tau = \Delta t \cdot \gamma  = \Delta t \cdot \sqrt{1 - (\frac{v}{c})^2}$
, } where $v$ is the velocity of an inertial observer following the geodesic between the two events and $\Delta \tau$ would then be the time that elapsed on their clock.
Light-like paths have a proper time interval equal to zero.
While a geodesic between two points in a Euclidean space is the shortest possible path, in a Minkowski space, conversely, a geodesic is the \textit{longest} possible path, i.e.: 
The path between two given events forwards in time with the longest possible proper time interval. $\mathcal{M}_\infty$ is a \textit{directed acyclic graph} (DAG).
Recall the geodesic distance on graphs that was a shortest path metric between vertices introduced in \mbox{Section \ref{shortest}} . 
Analogously we now use a geodesic distance on DAGs that is a \textit{longest} path metric between vertices. 
The length of such paths is determined by counting the number of steps, there are, however, multiple different possible methods for counting the number of steps for different types of GRIDS.
For GRIDS that consist of only time-like edges, one simply counts the number of steps along the time-like edges.
For GRIDS that consist of only light-like edges, one should count the number of direction changes of a path along light-like edges.
For GRIDS that consist of both light-like as well as time-like edges, such as our $\mathcal{M}_\infty$, we count the number of time-like edges while \textit{not} counting the light-like edges along a path that consists of both time-like as well as light-like edges.
In the following theorem, we show how accurately this geodesic distance approximates the formula of the proper time interval, 
for any $v$ that is slower than the speed of light $c=1$ \ by some arbitrarily small \mbox{constant $\delta$ .}

\begin{samepage}
\begin{theorem}\label{temporal}
{\textbf{Accuracy of the Proper Time Interval on $\mathcal{M}_\infty$}
}\\
For any arbitrarily small positive constant $\delta$ :\\
For all frame-grids $F_0$ of $\mathcal{M}_\infty$ :\\
Given two randomly selected vertices $Q$ and $U$ of $F_0$ :\\
Let\ $ s$ and $\Delta t$ be the spatial distance and the temporal distance between\\ the integer coordinates of $Q$ and $U$ on $F_0$ and let $\Delta\tau := \sqrt{(\Delta t)^2-s^2}$ .\\
Let \ $d(Q,U)$ be the geodesic distance between $Q$ and $U$\\ that is a \mbox{longest path metric} on $\mathcal{M}_\infty$ .
\\
\mbox{If \ $ \frac{s}{\Delta t} \leq 1-\delta $\ holds, then the following limiting behavior holds :}
\\
\[\text{relative error of the proper time interval} 
\ \ := \ \ \frac{\Delta\tau - d(Q,U)}{\Delta\tau}
\ \ = \ \ \mathcal{O}\left(\frac{\log \Delta \tau}{\Delta \tau}\right)\]
\end{theorem}
{\ }\end{samepage}
\begin{proof}
Recall that the proofs of \mbox{Theorem \ref{euclid}} and \mbox{Theorem \ref{luminal}} 
both involved the construction of a sufficiently short or quick path. 
Analogously, the proof of \mbox{Theorem \ref{temporal}} here would involve the construction of a sufficiently long path, 
i.e.: A path from the vertex $Q$ to the vertex $U$ forwards in time along sufficiently many time-like edges.
The quick path in the proof of \mbox{Theorem \ref{luminal}} contained \textit{three} approximately aligned straight \textit{light}-like paths that alternated with four paths that quickly transition between frame-grids. 
Analogously, the optimal long path here would contain \textit{four} approximately aligned straight \textit{time}-like paths that alternate with five paths that quickly transition between frame-grids.
This increase in numbers stems from the additional degree of freedom provided by the random selection of the time coordinates of $Q$ and $U$ , which was not present in back in \mbox{Theorem \ref{luminal}} .

Our path starts at $Q$ on $F_0$. 
Recall that $F_0$ , like all other frame-grids, corresponds to a point in the hyperbolic space. 
The initial part of our path leads from $F_0$ along interlacings across other frame-grids through hyperbolic space towards the vicinity of the point corresponding to the velocity vector from $Q$ to $U$ . 
This initial part of the path can be thought of as a quick acceleration nearing the desired velocity.
For the extreme cases where $Q$ and $U$ are selected such that the velocity $\frac{s}{\Delta t}$ is very close to the speed of light while the time $\Delta \tau$ is too minuscule to allow for the required acceleration, we introduced the constant $\delta$ that limits the velocities.
$\delta$ can be selected to be arbitrarily small, so long as it is positive and constant in order to allow for the error's limiting behavior (formulated in the theorem) when $\Delta \tau$ tends towards infinity.

The path continues through a tree that is analogous to the binary tree that we employed in the proof of \mbox{Theorem \ref{luminal}} , with the difference being that the nodes of the tree populate a 3D region instead of populating the spherical surface. 
The tree is used to in order to move to a frame-grid that is close to a desired velocity rather than being close to a desired orientation.
The rest of the proof is analogous to the proof of \mbox{Theorem \ref{luminal}} and is thus not described further here.
The proof technique described here also works to show the accuracy of the proper time interval on GRIDS consisting solely of time-like edges \mbox{as well as GRIDS consisting solely of light-like edges.}\\
\textbf{q.e.d.}
\end{proof}
$\ $

\subsubsection{Discussion of the Theorems}
For $\mathcal{M}_\infty$, our previous two theorems have shown the accuracy of the speed of light as well as the accuracy of the proper time interval for any inertial frame of reference. 
The \textit{proper length} of an object can be calculated from the time taken for a flash of light to travel back and forth along the object when measured by a clock in the rest frame of the object.
Therefore, the proper length $l_0$ measured on $\mathcal{M}_\infty$ inherits the same accuracy as the speed of light and the proper time interval, 
i.e.: The error of the proper length $l_0$ relative to itself is $ \mathcal{O}\left(\frac{\log l_0}{l_0}\right)$ .
Analogously to the geodesic distance in corollary \ref{voidOfCoordinates}, measuring the proper lengths between multiple spatial locations of the same inertial frame in $\mathcal{M}_\infty$ then yields exactly the values that one would expect from distances between points in a 3D Euclidean space, up to minuscule errors.
We hypothesize that the familiar 3D Euclidean space of our everyday reality is actually exactly this feature of a GRIDS.

Recall that Theorems \ref{luminal} and \ref{temporal} are concerned with paths that lead through multiple frame-grids but ultimately return to the same frame-grid that they started on, since these paths connect two vertices that belong to the \textit{same} frame-grid.
In some form, these theorems would also be valid for paths between two vertices that belong to the same finite subset of interlaced frame-grids rather than just the same single frame-grid.
However, the larger these subsets of frame-grids, the lower the resulting accuracies. 
On the other hand, this drop in accuracy does not occur if one is concerned with the geodesic path from a single vertex to a set of vertices that all reside at approximately the same location but are spread over all frame-grids.

\subsubsection{Quantum Superposition}
The following short thought experiment shows some emergent quantum properties of GRIDS.
Conway's famous so-called 'game of life' \cite{games1970fantastic} is a set of simple rules that determines how binary states on a 2D square grid change over time.
Conway's simple rules can generate a surprisingly diverse set of phenomena such as oscillating self-sustaining structures that can collide with each other and annihilate or produce other such structures.
Now imagine a modified version of this game that can be played on a GRIDS instead of a simple grid.
Note that two frame-grids that were not interlaced directly can have an arbitrarily low density of shared vertices, while being located arbitrarily close in the hyperbolic space.
Many different versions of the same oscillating self-sustaining structures could occupy the same space with almost no interference between them, due to the low density of shared nodes between arbitrary frame-grids,
This means that GRIDS, such as $\mathcal{M}_\infty$, possess the causal structure of quantum superpositions. 
Further investigation of these features of GRIDS is outside the scope of our paper.
Note that such quantum properties were a natural side effect of having searched for the simplest graphs that follow (3+1)-dimensional special relativity.
\\

\section{Conclusion}
The lengths of the shortest paths, or geodesic paths, on simple square lattices are non-Euclidean. 
It was thus often assumed that the same lack of isotropy would apply to all structures of similar simplicity and regularity.
However, in this paper we showed that similar structures can exhibit isotropy and even yield the full (3+1)-dimensional Minkowski spacetime, when measured using simple geodesic paths.
We call such structures GRIDS, which stands for Graphs that are Relativistic, Isotropic, Deterministic, and Simple. 
Due to Occam's razor, our key insight further increases the plausibility of the theories that assume spacetime to be a discrete structure, such as causal set theory, loop quantum gravity, and the Wolfram physics project.
In further theorems, we then demonstrated the accuracy of the speed of light as well as the accuracy of the proper time interval. 
Causal structures reminiscent of quantum superposition emerged as a side-effect.

Future research, in pursuit of low hanging fruit, should primarily be concerned with merging the GRIDS concept with theories that previously managed to derive aspects of general relativity \cite{einstein1915feldgleichungen} from discrete structures, such as Gorard's theory \cite{gorard2020some}.
While we provided simple rule sets that fully characterize examples of GRIDS, we did not provide the actual graph rewriting rules, which we recommend to be determined after the merging with the existing theories.
We hope that soon an even deeper understanding of the structure underlying spacetime will be gained, which shall allow for further steps towards the ultimate theory of fundamental physics.
\\
\\
\newpage
\bibliographystyle{plain}
\bibliography{references}

\begin{figure}[h]
     \centering
         \includegraphics[scale=0.57]{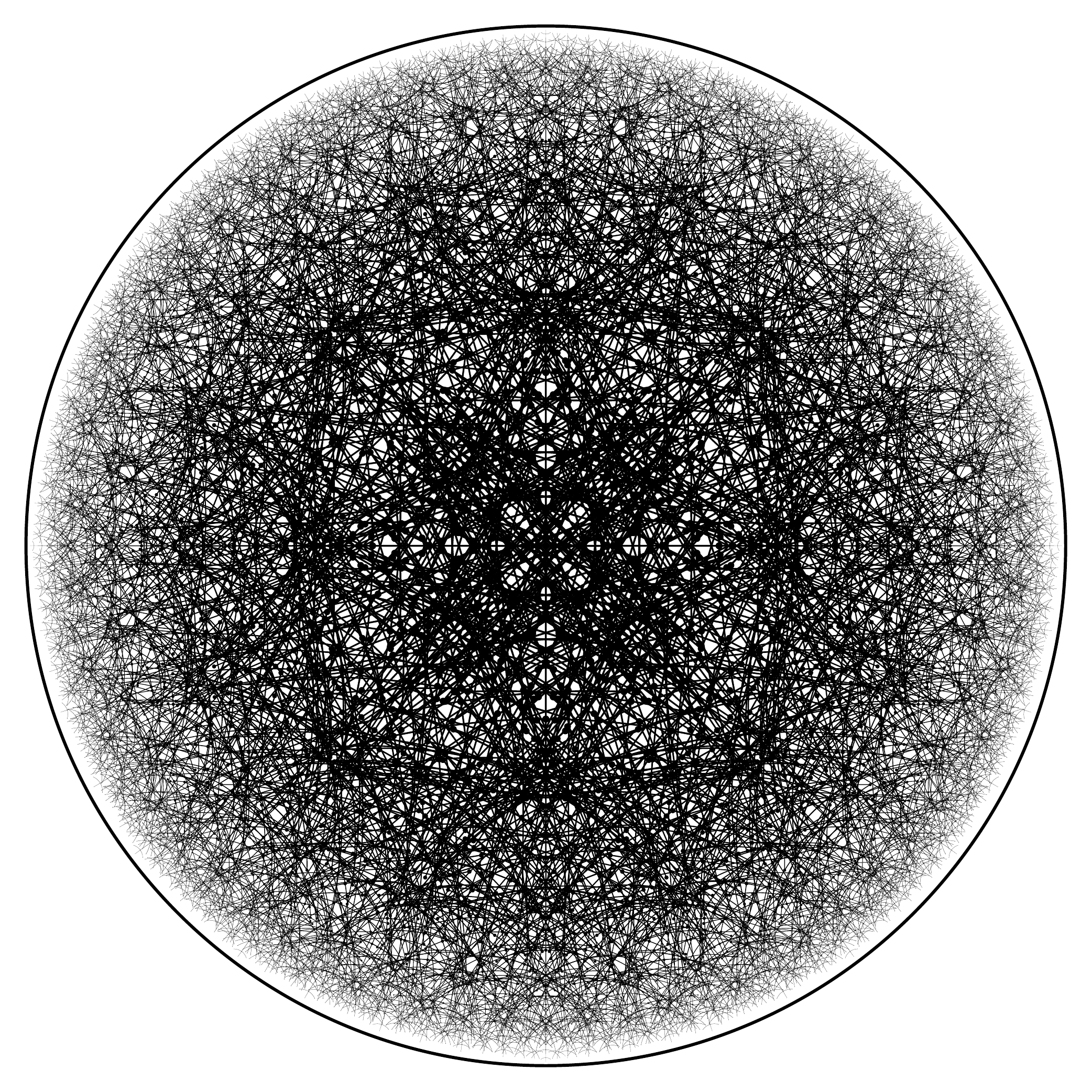}
         \caption{$\mathcal{M}_8$}
         \label{fig:y equals x}
\end{figure}

\end{document}